%% file: adiabatic-full-paper.tex
\documentclass[preprint,preprintnumbers,amsmath,amssymb,superscriptaddress]{revtex4}

\usepackage{graphicx}
\usepackage{bm} 
\usepackage{amsthm}
\usepackage[]{amsfonts, graphics, amsmath}

\renewcommand{\atop}[2]{\genfrac{}{}{0pt}{}{#1}{#2}}

\newtheorem{lemma}{Lemma}
\newtheorem{thm}{Theorem}

\newtheorem{claim}{Claim}
\newtheorem{fact}{Fact}

\newcommand{\N}{\mathbb{N}}
\newcommand{\x}{\mathbf{x}}
\newcommand{\y}{\mathbf{y}}

\newcommand{\mb}[1]{\boldsymbol{#1}} 

\newcommand{\ra}{\rangle}
\newcommand{\la}{\langle}

\newcommand{\be}{\begin{equation}}
\newcommand{\ee}{\end{equation}}
\newcommand{\ber}{\begin{eqnarray}}
\newcommand{\eer}{\end{eqnarray}}

\newcommand{\ket} [1] {\vert #1 \rangle}
\newcommand{\bra} [1] {\langle #1 \vert}
\newcommand{\braket}[2]{\langle #1 | #2 \rangle}
\newcommand{\proj}[1]{\ket{#1}\bra{#1}}
\newcommand{\id}{\mb{1}}
\newcommand{\norm}[1]{||#1||}
\newcommand{\mean}[1]{\left\langle #1 \right\rangle}
\newcommand{\var}[1]{\sigma^2\left( #1 \right)}

\newcommand{\typ}[1]{{#1}_{\mathrm{typ}}}

\begin{document}

\title{Adiabatic quantum optimization fails for random instances of NP-complete problems}
\author{\sc Boris Altshuler}\email{bla@phys.columbia.edu}
\affiliation{Columbia University}
\affiliation{NEC Laboratories America Inc.}
\author{\sc Hari Krovi}\email{hari.krovi@uconn.edu}
\affiliation{NEC Laboratories America Inc.}
\author{\sc Jeremie Roland}\email{jroland@nec-labs.com}
\affiliation{NEC Laboratories America Inc.}
\date{\today}

\begin{abstract}
Adiabatic quantum optimization has attracted a lot of attention because small scale simulations gave hope that it would allow to solve NP-complete problems efficiently. Later, negative results proved the existence of specifically designed hard instances where adiabatic optimization requires exponential time. In spite of this, there was still hope that this would not happen for random instances of NP-complete problems. This is an important issue since random instances are a good model for hard instances that can not be solved by current classical solvers, for which an efficient quantum algorithm would therefore be desirable. Here, we will show that because of a phenomenon similar to Anderson localization, an exponentially small eigenvalue gap appears in the spectrum of the adiabatic Hamiltonian for large random instances, very close to the end of the algorithm. This implies that unfortunately, adiabatic quantum optimization also fails for these instances by getting stuck in a local minimum, unless the computation is exponentially long.
\end{abstract}

\maketitle

\section{Introduction}
Adiabatic quantum computing is a computational paradigm (introduced in~\cite{farh00}) where the solution to an optimization problem is encoded in the ground state of some Hamiltonian $H_P$. An adiabatic algorithm would proceed as follows: prepare the ground state of another Hamiltonian $H_0$ (chosen so that its ground state is easy to prepare), then slowly modify the Hamiltonian of the system from $H_0$ to $H_P$, using an interpolation $H(s)=(1-s)H_0 +s H_P$. If this is done slowly enough, the Adiabatic Theorem of Quantum Mechanics~\cite{messiah} ensures that the system will stay close to the ground state of the instantaneous Hamiltonian throughout the evolution, so that we finally obtain a state close to the ground state of $H_P$. At this point, measuring the state will give us the solution of our problem with high probability. To put this in more quantitative terms, if the problem size (number of bits) is $N$, then one requires the instantaneous eigenvalue gap $\Delta(s)$ between the ground state and the first excited state to be inverse polynomial in $N$ at each step $s$ of the evolution. The computation time $T$ scales as the inverse square of the gap $T\sim 1/\Delta^2$, where $\Delta=\min_s \Delta(s)$. This means that if the eigenvalue gap becomes exponentially small at any point in the evolution, then the computation requires exponential time. This eigenvalue gap provides a possible advantage of adiabatic quantum computing compared to the usual model based on quantum circuits. Since the system stays in its ground state throughout the evolution, robustness against thermal noise and decoherence could be provided by the eigenvalue gap~\cite{chil01,rc05,Lidar08}. It was also shown that adiabatic quantum computing is universal for quantum computing~\cite{avkll04}, i.e., any algorithm expressed as a quantum circuit may be translated into an adiabatic algorithm, and vice versa.

Adiabatic quantum computing was first proposed as a new approach to solve hard optimization problems, and has attracted a lot of attention because numerical evidence~\cite{fggllp01} seemed to indicate that the time required to solve NP-complete problems scaled only polynomially with the problem size, at least for small sizes. But later work gave strong evidence that this may not be the case. Refs.~\cite{Farhi_fail,Farhi_SA} show that adiabatic algorithms can fail if one does not choose the initial Hamiltonian carefully by taking into account the structure of the problem. In Ref.~\cite{Reich} the Hamiltonians of certain instances of 3-SAT were mapped to an Ising model and diagonalized analytically. It was shown that the gap is exponentially small in some cases. Refs.~\cite{vdam01} and \cite{vdam2} construct special instances of 3-SAT which are hard for the adiabatic algorithm to solve. More recently, it was shown that very small gaps could appear in the spectrum of the Hamiltonian due to an avoided crossing between the ground state and another level corresponding to a local minimum of the optimization problem~\cite{AminChoi,mit}.
However, these results show the failure of adiabatic quantum optimization for specifically designed hard instances. In this paper, we show that adiabatic quantum optimization fails with high probability for randomly generated instances of the NP-complete problem Exact Cover 3 (EC3), also known as 1-in-3 SAT. Since the core of the argument leading to this conclusion only relies on general properties shared by other NP-complete problems such as 3-SAT, this provides a strong evidence that adiabatic quantum optimization typically cannot solve hard instances of NP-complete problems efficiently.

Our argument relies on different elements. In Section~\ref{sec:preliminaries}, we introduce the problem Exact Cover 3 and the adiabatic algorithm proposed to solve it. In Section~\ref{sec:statistics}, we study some statistical properties of random instances of EC3, which will be crucial to the result.  In Section~\ref{sec:perturbation}, we study the perturbation expansion of eigenenergies of the adiabatic Hamiltonian.  In Section~\ref{sec:avoided-crossing}, we show that perturbation theory predicts an avoided crossing and therefore a small gap occurring close to the end of the adiabatic evolution, where perturbation theory becomes valid. We also performed numerical simulations to confirm the predictions of perturbation theory and estimate the position of the avoided crossing. Finally, in Section~\ref{sec:gap-scaling}, we show that this avoided crossing induces an exponentially small gap. Since these results rely on perturbation expansions, we discuss in Section~\ref{sec:applicability} the applicability of perturbation theory. We show how this problem is intimately related to the phenomenon of Anderson localization~\cite{anderson58}, which implies that the eigenstates of the Hamiltonian are localized for small perturbation, corresponding in our case to the end of the algorithm, close to $s=1$. An important observation is that the adiabatic Hamiltonian $H(s)$ has exactly the same form as the model used by Anderson to describe localization, except that the particle evolves on an $N$-dimensional hypercube instead of a $d$-dimensional lattice. This emphasizes the relevance of Anderson localization for the study of quantum algorithms, even though it is currently not commonly used in this context, one example being Ref.~\cite{keating} where it has been used to show weaknesses in quantum walk algorithms.

\section{Preliminaries}\label{sec:preliminaries}
\subsection{Exact Cover}
Exact Cover 3 (EC3) is an NP-complete problem which was considered for adiabatic algorithms in~\cite{fggllp01}. Consider an $N$-bit string  $\x=(x_1,x_2,\dots ,x_N)$, where $x_i\in\{0,1\}$. An instance of EC3 consists of many clauses each containing three bits. A clause $C=(x_{i_C},x_{j_C},x_{k_C})$ is said to be satisfied if and only if one of the three bits is one and the other two are zero, i.e., $x_{i_C}+x_{j_C}+x_{k_C}=1$. A solution is an assignment of the bits which satisfies all the clauses. The solutions of an instance of EC3 can be encoded into a cost function given by
\be\label{eq:cost-function}
f(\x)=\sum_C (x_{i_C}+x_{j_C}+x_{k_C}-1)^2.
\ee
A solution is therefore an assignment which yields a zero cost.

We will call \emph{random instances} of EC3 with $N$ bits and $M$ clauses those generated by picking uniformly at random $M$ clauses of three bits, with replacement. This distribution of instances is important since it generates hard instances that can not be solved by current classical solvers. The hardness of random instances depends highly on the clauses-to-variables ratio $\alpha=\frac{M}{N}$. As $\alpha$ increases, we observe two phase transitions~\cite{biroli,zdeborova}. For low $\alpha$, the density of the solutions is high and essentially uniform. As $\alpha$ increases to the \emph{clustering threshold} $\alpha_c$, the first phase transition occurs where the problem goes from having many solutions to clustered solutions with different clusters isolated from each other. The other phase transition occurs at the \emph{satisfiability threshold} $\alpha_s$ after which the problem is unsatisfiable with high probability. Therefore, the hard instances with only a few isolated solutions, which cannot be solved efficiently by known classical algorithms, lie just before this second phase transition, and this is the regime we will be interested in. The satisifiability threshold of EC3 has been studied in~\cite{raymond07}, where it is shown that $\alpha_s=0.6263\pm 10^{-4}$.


\subsection{Adiabatic Algorithm for Exact Cover}
To design an adiabatic quantum algorithm for this problem, we build a problem Hamiltonian $H_P$ acting on a space of $N$ qubits such that each state $\ket{\x}$ of the computational basis is an eigenstate of $H_P$ with energy $E_\x=f(\x)$.
Using the mapping $x_i\to(1-\sigma_z^{(i)})/2$ in Eq.~(\ref{eq:cost-function}), where $\sigma_z^{(i)}$ is the Pauli operator $\sigma_z$ acting on the $i$-th bit, we get the following expression for the problem Hamiltonian
\ber\label{eq:hp}
H_P&=&M\id-\frac{1}{2}\sum_{i=1}^N B_i\sigma_z^{(i)} + \frac{1}{4}\sum_{i,j=1}^N J_{ij}\sigma_z^{(i)}\sigma_z^{(j)},
\eer
where $\id$ is the identity operator, $M$ is the total number of clauses, $B_i$ is the number of clauses in which the bit $i$ participates and $J_{ij}$ is the number of clauses where the bits $i$ and $j$ participate together (so $J_{ij}=J_{ji}$ and we set $J_{ii}=0$ for convenience). The solution to the EC3 instance is now given by the ground state of $H_P$. As for the initial Hamiltonian $H_0$, a standard choice is
\be\label{eq:h0}
H_0=-\sum_{i=1}^N \sigma_x^{(i)} ,
\ee
where $\sigma_x^{(i)}$ is the Pauli operator $\sigma_x$ on the $i$-th bit. $H_0$ is therefore a 1-local Hamiltonian accepting as unique ground state the uniform superposition
\be
\ket{\psi_0}=2^{-\frac{N}{2}}\sum_{\x\in\{0,1\}^N}\ket{\x}.
\ee
The adiabatic quantum algorithm consists in preparing $\ket{\psi_0}$, applying the Hamiltonian $H_0$ and slowly modifying the system Hamiltonian, following an interpolation
\be
H(s(t))=(1-s(t))H_0+s(t)H_P ,
\ee
where $s(t)=t/T$ and $T$ is the computation time. Let $\Delta(s)$ be the eigenvalue gap between the ground state and the first excited state of $H(s)$. If $T$ is large compared to $1/\Delta^2$, where $\Delta=\min_s \Delta(s)$, we will obtain a state close to the ground state of $H_P$ at the end of the evolution, so that a measurement in the computational basis will yield the solution to the problem with high probability. To evaluate the complexity of this algorithm, we therefore need to find the minimum gap of $H(s)$.

\section{Statistical properties of random instances}\label{sec:statistics}
Since we are are interested in random instances of EC3, our results will rely on some statistical properties of such instances. Recall that a random instance is obtained by picking uniformly and independently $M$ clauses of $3$ bits among a set of $N$ bits. Let us study the statistical properties of such instances in the limit of large $N$, for fixed clauses-to-variables ratio $\alpha=\frac{M}{N}$.

From Eq.~(\ref{eq:hp}), we see that an EC3 instance over $N$ bits is completely specified by the $N\times N$ matrix $(J_{ij})_{i,j=1}^N$ (note that $B_i=\frac{1}{2}\sum_j J_{ij}$ and $M=\frac{1}{3}\sum_i B_i$). Such a matrix defines a graph $G$ over $N$ vertices such that there is an edge $(i,j)$ if and only if $J_{ij}\neq 0$. We will now show that for random instances, the local properties of $G$ are independent of $N$.

Let us first focus on the degree of the graph. Since each clause involves two other bits, the degree of vertex $i$ is at most twice $B_i$, the number of clauses involving bit $i$. Since the probability that bit $i$ appears in one random clause is $3/N$, and the $M$ clauses are picked uniformly at random and with replacement, $B_i$ follows a binomial distribution
\be\label{eq:proba-b}
\Pr[B_i=b]=\binom{M}{b}\left(\frac{3}{N}\right)^b\left(1-\frac{3}{N}\right)^{M-b}.
\ee
In the limit $N\to\infty$, for fixed $b$ and $\alpha=\frac{M}{N}$, we have 
\be\label{eq:lim-proba-b}
\lim_{N\to\infty}\Pr[B_i=b]=e^{-3\alpha}\frac{(3\alpha)^b}{b!}.
\ee
The fact that this distribution converges for large $N$, as well as other properties of the graph $G$, will be crucial to our results. In particular, the fact that $B_i$ follows the binomial distribution in Eq.~(\ref{eq:proba-b}) immediately implies the following (we denote by $\mean{V}$ and $\var{V}$ the mean value and variance of a random variable $V$).
\begin{fact}\label{fact:proba-b}
For random EC3 instances with $\alpha=\frac{M}{N}$, we have $\mean{B_i}=3\alpha$ and $\var{B_i}=3\alpha(1-\frac{3}{N})$.
\end{fact}
From Markov's inequality, this implies that $B_i$, and in turn the degree of each vertex, remains bounded with high probability in the limit $N\to\infty$.

From Eq.~(\ref{eq:lim-proba-b}), we also see that $\lim_{N\to\infty}\Pr[B_i=0]=e^{-3\alpha}$, so that a given bit will not appear in any clause with probability $e^{-3\alpha}$. This implies that when we generate a random instance with $N$ bits, only a fraction of the bits will actually play a role in the instance. For a given random instance, let $N'$ be the number of bits present in some clause. From Eq.~(\ref{eq:proba-b}), we can show that the fraction of present bits $N'/N$ becomes more and more peaked around its mean value $1-e^{-3\alpha}$ in the limit $N\to\infty$.
\begin{fact}
 $\lim_{N\to\infty} \frac{\mean{N'}}{N}=1-e^{-3\alpha}$ and $\lim_{N\to\infty} \frac{\var{N'}}{N}=e^{-3\alpha}(1-e^{-3\alpha})$.
\end{fact}

For any set of bits $S\subseteq[N]$, let us define the induced subgraph $G_S$ as the graph on the set of vertices $S$ such that $(i,j)\in S\times S$ is an edge of $G_S$ if and only if it is also an edge of $G$. When there is no ambiguity, we will sometimes use $S$ to denote the subgraph itself. For a given EC3 instance on $N$ bits, we denote by $\mathcal{G}_u$ the set of subsets $S\subseteq [N]$ of size $u$ whose associated subgraphs are connected, or in short the set of connected graphs of size $u$. Let $G_u=|\mathcal{G}_u|$ be the number of connected graphs of size $u$. We will later use the fact that $G_u$ is linear in $N$ (the proof is given in Appendix~\ref{app:number-graphs}).
\begin{lemma}\label{lem:mean-size-connected}
 For any $u\in\mathbb{N}$, we have $\mean{G_u}=\Theta(N)$.
\end{lemma}


\section{Perturbation theory for the adiabatic Hamiltonian}\label{sec:perturbation}
\subsection{Perturbation theory using Green's functions}\label{sec:green-function}
In the following sections, we will show that the Hamiltonian $H(s)$ exhibits an exponentially small gap close to $s=1$. To study the spectrum of $H(s)$ around $s=1$, let us consider the Hamiltonian $H(\lambda)=\frac{H(s)}{s}=H_P+\lambda V$, where $\lambda=\frac{1-s}{s}$ and $V=H_0$ acts as a time independent perturbation on $H_P$. We describe how the spectrum of this Hamiltonian can be written as a perturbation expansion in powers of $\lambda$. Let $\ket{\x}$ be a non-degenerate eigenstate of $H_P$ with energy $E_\x$. We define the self-energy as
\be\label{eq:self-energy}
\Sigma_\x(E)=\sum_{q=1}^\infty \lambda^q \Sigma_\x^{(q)}(E),
\ee
where
\be
\Sigma_\x^{(q)}(E)=\sum_{\y^1,\dots ,\y^{q-1}}\frac{V_{\x \y^1}V_{\y^1\y^2}\dots V_{\y^{q-1}\x}}{(E-E_{\y^1})(E-E_{\y^2})\dots (E-E_{\y^{q-1}})},
\ee
$V_{\y^i\y^j}=\la \y^i|V|\y^j\ra$, and the sum in the last expression is over all eigenstates of $H_P$ different from $\x$. The perturbed eigenvalue $E_\x(\lambda)$ is then given by the pole of the Green's function
\be
G_\x(E)=\frac{1}{E-E_\x-\Sigma_\x(E)}.
\ee
Therefore, a perturbation expansion
\be\label{Eig_expansion}
E_\x(\lambda)=E_\x+\sum_{q=1}^\infty \lambda^q E_\x^{(q)}
\ee
 may be obtained by solving the equation $E=E_\x+\Sigma_\x(E)$ recursively using the perturbation expansion~(\ref{eq:self-energy}) for the self-energy.
The energy up to first order is
\be
E_\x(\lambda)
=E_\x+\lambda\Sigma_\x^{(1)}(E_\x)+O(\lambda^2),
\ee
so that $E_\x^{(1)}=\Sigma_\x^{(1)}$, where, when not explicitly written, all $\Sigma_\x^{(q)}(E)$ (and later also their derivatives) are evaluated at $E=E_\x$.
The energy up to second order term is
\ber
E_\x(\lambda)&=&E_\x+\lambda\Sigma_\x^{(1)}(E_\x+\lambda\Sigma_\x^{(1)})+\lambda^2\Sigma_\x^{(2)}(E_\x)+O(\lambda^3)\nonumber \\ &=&E_\x+\lambda\Sigma_\x^{(1)}+\lambda^2\left(\Sigma_\x^{(1)\prime}\Sigma_\x^{(1)}+\Sigma_\x^{(2)}\right)+O(\lambda^3) ,
\eer
where we have used Taylor series expansion and kept terms up to second order in $\lambda$. The second order correction $E_\x^{(2)}$ is then given by the coefficient of $\lambda^2$.
Similarly, the energy up to third order is given by
\ber
E_\x(\lambda)&=&E_\x
+ \lambda\Sigma_\x^{(1)}\left(E_\x+\lambda\Sigma_\x^{(1)}+\lambda^2\left(\Sigma_\x^{(1)\prime}\Sigma_\x^{(1)}+\Sigma_\x^{(2)}\right)\right)\nonumber\\
&&
+\lambda^2\Sigma_\x^{(2)}(E_\x+\lambda\Sigma_\x^{(1)})
+ \lambda^3\Sigma_\x^{(3)}(E_\x)+O(\lambda^4) \nonumber \\
&=& E_\x + \lambda \Sigma_\x^{(1)} + \lambda^2 \left(\Sigma_\x^{(1)\prime}\Sigma_\x^{(1)}+\Sigma_\x^{(2)}\right)\nonumber\\
&&
+ \lambda^3\left( \frac{1}{6}((\Sigma_\x^{(1)})^3)^{\prime\prime}+(\Sigma_\x^{(1)}\Sigma_\x^{(2)})^\prime + \Sigma_\x^{(3)}\right)+O(\lambda^4).
\eer

We can now take $V=H_0=-\sum_i\sigma_x^{(i)}$ and give the first few orders of the expansion. First note that in this case $\Sigma_\x^{(q)}(E)=0$ for every odd $q$. This is because $\la \x|H_0|\y\ra\neq 0$ if and only if $\x$ and $\y$ differ by one bit. 
Since at least one $\Sigma_\x^{(q')}$ (or some derivative of it) of odd order $q'$ will appear in each term of the correction $E_\x^{(q)}$ for odd order $q$, all odd orders in the perturbation expansion vanish. In this case, the first three non-zero terms are given below.
\ber
E_\x^{(2)}&=&\Sigma_\x^{(2)} \\
E_\x^{(4)}&=&\frac{1}{2}((\Sigma_\x^{(2)})^2)^\prime + \Sigma_\x^{(4)} \label{eq:energy-order-4} \\
E_\x^{(6)}&=&\frac{1}{6}((\Sigma_\x^{(2)})^3)^{\prime\prime}+(\Sigma_\x^{(2)}\Sigma_\x^{(4)})^\prime + \Sigma_\x^{(6)} .
\eer
We would like to express each of these corrections as a sum over paths going from the assignment $\x$ and back. Consider
\be
E_\x^{(2)}=\Sigma_\x^{(2)}=\sum_{\y} \frac{\la \x|H_0|\y\ra\la \y|H_0|\x\ra}{E_\x-E_\y}\label{eq:energy-order-2}.
\ee
Since the only non-zero terms arise when $\y$ is a single bit flip away from $\x$, we can think of $E_\x^{(2)}$ as a sum over all paths going from the assignment $\x$ and back which consist in flipping (and flipping back) only one bit. Similarly, we can think of $E_\x^{(q)}$ as a sum over all paths on the hypercube which consist in flipping any $q/2$ bits and flipping them back in all possible sequences. Thus, we define $A(P)$ such that
\be\label{eq:corrections}
E_\x^{(q)}=\sum_{P:\sum_ip_i=q/2}A(P) ,
\ee
where $P=(p_i)_{i=1}^N\in\N^N$ is a vector whose $i$-th component specifies half the number of times bit $i$ is flipped (we take {\it half} the number since any bit that is flipped  must be flipped back.) Of course, specifying $P$ does not uniquely specify a path, so $A(P)$ involves a sum over all paths corresponding to $P$.

\subsection{Scaling of corrections at successive orders}
In this section, we will show that when evaluating eigenvalues of the Hamiltonian $H(\lambda)$ around $\lambda=0$ by perturbation theory as described in the previous subsection, corrections for successive orders are all of order $\Theta(N)$.  Since all corrections are of the same order, this means that for large $N$, the leading behavior is given by the first non-zero correction in the expansion. This also suggests that the range of $\lambda$ for which the leading order in the perturbation expansion gives an accurate approximation is $N$-independent (this statement will be discussed in Section~\ref{sec:applicability}).

Let us consider the corrections $E_\x^{(q)}$ in Eq.~(\ref{eq:corrections}). We denote as $S(P)\subseteq[N]$ the set of bits that are flipped at least once in the paths corresponding to $P$, i.e., $S(P)=\{i\in[N]:p_i>0\}$. By extension, $S(P)$ then also defines, via the matrix $(J_{ij})$, a graph where vertices correspond to elements of $S(P)$.
In order to show that all corrections $E_\x^{(q)}$ scale as $\Theta(N)$, we prove that we do not need to consider all vectors $\{P:\sum_ip_i=q/2\}$ but only those associated to connected graphs $S(P)$ (we say that the graph is disconnected if $S(P)$ can be expressed as a disjoint union $S_1\cup S_2$ such that $J_{ij}=0$ for all $i\in S_1$ and $j\in S_2$).
\begin{lemma}\label{lem:correction-disconnected}
Let $P_0\in\N^N$ be such that the graph associated to $S(P_0)$ is disconnected. Then, $A(P_0)=0$.
\end{lemma}
\begin{proof}
Let $\vec{\lambda}=(\lambda_1,\dots ,\lambda_N)$ denote a multi-dimensional perturbation parameter, and let us consider the following generalized Hamiltonian
\be
H'(\vec{\lambda})=M\id-\frac{1}{2}\sum_{i=1}^N B_i\sigma_z^{(i)} +\frac{1}{4}\sum_{i,j=1}^N J_{ij}\sigma_z^{(i)}\sigma_z^{(j)} -\sum_{i=1}^N\lambda_i\sigma_x^{(i)}.
\ee
It can be seen that the perturbation expansion of the eigenvalues of this Hamiltonian can be written as
\be\label{E_tilde}
E'_\x(\vec{\lambda})=\sum_{q=0}^\infty\sum_{P:\sum_ip_i=q/2} \!\!\!A(P)\prod_i\lambda_i^{2p_i},
\ee
with the same coefficients $A(P)$ as for $E_\x(\lambda)$.

Now, let $S=S(P_0)$, and consider the Hamiltonian obtained from the generalized Hamiltonian by substituting $\lambda_j=0$ if $j\notin S$,
\be
H'(\vec{\lambda}_{S})=M\id-\frac{1}{2}\sum_{i=1}^N B_i\sigma_z^{(i)} +\frac{1}{4}\sum_{i,j=1}^N J_{ij}\sigma_z^{(i)}\sigma_z^{(j)} -\sum_{i\in S}\lambda_i\sigma_x^{(i)} ,
\ee
where $\vec{\lambda}_{S}$ is the vector obtained from $\vec{\lambda}$ by performing this substitution.
It is easy to see that the perturbation expansion of the eigenvalue corresponding to assignment $\x$ is given by
\be\label{eq:eigenvalue-reduced}
E'_{\x}(\vec{\lambda}_{S})=\sum_{q=0}^\infty\sum_{P:\left\{\atop{\sum_ip_i=q/2}{S(P)\subseteq S}\right.} \!\!\!A(P)\prod_i\lambda_i^{2p_i} ,
\ee
again with the same coefficients $A(P)$ as above. Now, observe that the operators $\sigma_z^{(i)}$ for $i\notin S$ commute with the Hamiltonian $H'(\vec{\lambda}_{S})$. Therefore,
the bits outside of $S$ fall out of the dynamics, and it suffices to study the Hamiltonian obtained from $H'(\vec{\lambda}_{S})$ by substituting $\sigma_z^{(i)}$ with $(-1)^{x_j}$, where $x_j$ is the value of the $j$-th bit in assignment $\x$, that is,
\be\label{eq:reduced-hamiltonian}
H'_{S}(\vec{\lambda}_{S})=-\frac{1}{2}\sum_{i\in S} B'_i\sigma_z^{(i)} +\frac{1}{4}\sum_{i,j\in S} J_{ij}\sigma_z^{(i)}\sigma_z^{(j)} -\sum_{i\in S}\lambda_i\sigma_x^{(i)}
\ee
where
\ber
B'_i&=&B_i-\frac{1}{2}\sum_{j\notin S}J_{ij}(-1)^{x_j},
\eer
and we have ignored an irrelevant term proportional to $\id$.
The eigenvalue $E'_{S,\x}(\vec{\lambda}_{S})$ of this Hamiltonian then coincides with that of Hamiltonian $H'(\vec{\lambda}_{S})$, given in Eq.~(\ref{eq:eigenvalue-reduced}), up to this irrelevant constant.

By assumption, the graph associated to $S$ is disconnected, so there exist disjoint non-empty sets $S_1$ and $S_2$ such that $S=S_1\cup S_2$ and $J_{ij}=0$ for all $i\in S_1$ and $j\in S_2$.
Therefore, we can write $H'_S(\vec{\lambda}_{S})$ as 
\be\label{eq:sum-hamiltonian}
H'_{S}(\vec{\lambda}_{S})=H'_{S_1}(\vec{\lambda}_{S_1})+H'_{S_2}(\vec{\lambda}_{S_2}).
\ee
The perturbation expansion of the eigenvalue $E'_{S_k,\x}(\vec{\lambda}_{S_k})$ of Hamiltonian $H'_{S_k}$ can be written similarly as above
\be\label{eq:reduced-expansion}
E'_{S_k,\x}(\vec{\lambda}_{S_k})=\sum_{q=0}^{\infty}\sum_{P:\left\{\atop{\sum_ip_i=q/2}{S(P)\subseteq S_k}\right.}\!\!\!A(P)\prod_i\lambda_i^{2p_i} .
\ee
Moreover, the Hamiltonians $H'_{S_1}(\vec{\lambda}_{S_1})$ and $H'_{S_2}(\vec{\lambda}_{S_2})$ commute since they only act non-trivially on different qubits, so by Eq.~(\ref{eq:sum-hamiltonian}) we have
\be
E'_{S,\x}(\vec{\lambda}_{S})= E'_{S_1,\x}(\vec{\lambda}_{S_1})+E'_{S_2,\x}(\vec{\lambda}_{S_2}).
\ee
Since there are no terms proportional to $\prod_i\lambda_i^{p_i}$ in the expansion~(\ref{eq:reduced-expansion}) for any $P=(p_i)$ such that $S(P)=S$, we must have that $A(P)=0$ for any such $P$.
\end{proof}


We now show that for connected graphs, the coefficients are finite.
\begin{lemma}\label{lem:correction-connected}
 Let $P_0\in\N^N$ such that the graph associated to $S(P_0)$ is connected and of size $u=O(1)$. Then, $\mean{A(P_0)}=\Theta(1)$.
\end{lemma}
\begin{proof}
 Let $S=S(P_0)$. Then, for any $P$ such that $S(P)\subseteq S$, the amplitude $A(P)$ in the perturbation expansion of the eigenvalue of $H$ is the same as for the Hamiltonian $H'_S(\vec{\lambda}_{S})$ in Eq.~(\ref{eq:reduced-hamiltonian}). From $\sum_{j} J_{ij}=2B_i$ and the fact that for any $i\in S$, there exists $j\in S$ such that $J_{ij}\geq 1$ (which follows from the fact that $S$ is connected), we also see that $\frac{1}{2}\leq B'_i\leq 3B_i$. From Fact~\ref{fact:proba-b}, this implies that $\mean{B'_i}=\Theta(1)$. Since the Hamiltonian $H'_S$ acts on a finite number of bits $u$, the perturbation expansion of its eigenvalues is $N$ independent, which proves the lemma.
\end{proof}

We may now prove the following theorem.
\begin{thm}\label{thm:corrections-order-n}
For any $q=O(1)$, the $q$-th order correction $E_\x^{(q)}$ of an eigenvalue of Hamiltonian $H$ scales as $\mean{E_\x^{(q)}}=\Theta(N)$.
\end{thm}
\begin{proof}
From Lemma~\ref{lem:correction-disconnected}, this $q$-th order correction may be written as
\be
E_\x^{(q)}= \sum_{P:
\left\{\atop{\sum_ip_i=q/2}{S(P)\ {\rm connected}} \right.} A(P).
\ee
From Lemma~\ref{lem:mean-size-connected}, the number of terms in this sum is $\Theta(N)$ on average, and from Lemma~\ref{lem:correction-connected}, each of these terms is $\Theta(1)$ on average, which implies the theorem.
\end{proof}


\section{Avoided crossing}\label{sec:avoided-crossing}
\subsection{General idea}
In this section, we will show that for random instances of EC3, perturbation theory predicts that the spectrum of the Hamiltonian $H(s)$ will exhibit an avoided crossing, and therefore a small eigenvalue gap, close to $s=1$. The general strategy will be the following. We first consider an instance of EC3 with at least two satisfying assignments which are isolated i.e., the Hamming distance between the solutions (the number of bits in which the two solutions differ) is of order $\Theta(N)$.
Then, we modify the instance by adding a clause which is satisfied by one of the solutions, but not by the other, which will now correspond to a local minimum of the new cost function. We show that this can create an avoided crossing between the levels corresponding to the solution and the local minimum, and therefore a small gap in the spectrum of $H(s)$.

\subsection{Analysis by perturbation theory}
As detailed in Section~\ref{sec:green-function}, perturbation theory allows to evaluate the energy $E_\x(\lambda)$ of an eigenstate corresponding to an assignment $\x$ as
$$
E_\x(\lambda)=E_\x+\sum_{q=1}^\infty \lambda^{(q)}\ E_\x^{(q)}.
$$

For two solutions $\x^{1},\x^{2}$, let $E_{12}(\lambda)=E_{1}(\lambda)-E_{2}(\lambda)$ be the splitting between the two energies, and $E_{12}^{(q)}$ be the $q$-th order correction to this splitting. We know from Section~\ref{sec:perturbation} that the first non-zero correction to $E_\x(\lambda)$ yields the leading behavior even as $N$ increases since all corrections scale as $\Theta(N)$. Moreover, recall from Section~\ref{sec:green-function} that all odd order corrections are zero so that the first non-zero correction to $E_\x(\lambda)$ appears at second order. From Eq.~(\ref{eq:energy-order-2}), we have
\ber
E_\x^{(2)}&=&\sum_{\y} \frac{\la \x|H_0|\y\ra\la \y|H_0|\x\ra}{E_\x-E_\y}\\
&=&-\sum_{i=1}^N \frac{1}{B_i}
\eer
where we used the fact that $E_\x=0$ since $\x$ is a solution, and that $E_\y=B_i$ if $\y$ differs from $\x$ only in the $i$-th bit, since assignment $\y$ will violate all the clauses where this bit appears. Since this correction does not depend on the particular solution $\x$ we start from, we have $E_{12}^{(2)}=0$, so the correction to the splitting $E_{12}(\lambda)$ between two solutions can only appear at order 4. Eq.~(\ref{eq:energy-order-4}) yields
\be\label{eq:energy-order-4-ec3}
E_\x^{(4)} =
\sum_{i=1}^N \frac{1}{B_i^3}
+\frac{1}{2}
\sum_{\atop{i,j=1}{i\neq j}}^N \left(\frac{1}{B_i} + \frac{1}{B_j} \right)^2 \left(\frac{1}{B_i+B_j}-\frac{1}{E^\x_{ij}}\right)
\ee
where $E^\x_{ij}=E_\y$ for the assignment $\y$ obtained from the solution $\x$ by flipping bits $i$ and $j$~\footnote{We have $E^\x_{ij}\neq 0$ unless there is another solution at Hamming distance 2, which would typically not happen for an instance that has been properly cleaned as described in Appendix~\ref{app:numerical}.}.  When bits $i$ and $j$ never appear together in a clause, i.e., $J_{ij}=0$, we have $E^\x_{ij}=B_i+B_j$ and the corresponding term in Eq.~(\ref{eq:energy-order-4-ec3}) is zero. Therefore, we only need to sum over $i,j$ such that $J_{ij}\neq 0$, that is,
\be
E_\x^{(4)} =
\sum_{i=1}^N \frac{1}{B_i^3}
+\frac{1}{2}
\sum_{\atop{i,j=1}{J_{ij}\neq 0}}^N \left(\frac{1}{B_i} + \frac{1}{B_j} \right)^2 \left(\frac{1}{B_i+B_j}-\frac{1}{E^\x_{ij}}\right),
\ee
so that this expression only involves $\Theta(N)$ terms, as was shown in Section~\ref{sec:perturbation}. The first non-zero correction for the splitting $E_{12}(\lambda)$ is then given by
\be
E_{12}^{(4)} =
\frac{1}{2}
\sum_{\atop{i,j=1}{J_{ij}\neq 0}}^N \left(\frac{1}{B_i} + \frac{1}{B_j} \right)^2 \left(\frac{1}{E^{\x^{2}}_{ij}}-\frac{1}{E^{\x^{1}}_{ij}}\right),
\ee
so that
$
E_{12}(\lambda)=\lambda^4\ E_{12}^{(4)}+O(\lambda^6).
$
For random instances, $B_i$, $B_j$ and $E^{\x}_{ij}$ will become random numbers. Since $E_{12}^{(4)}$ is given by a sum of $\Theta(N)$ random terms with zero mean, we can expect the variance of $E_{12}^{(4)}$ to be of order $\Theta(N)$,
\be\label{linear}
\mean{(E_{12}^{(4)})^2}\approx C^{(4)} N,
\ee
and similarly for the $p$-th percentiles,
\be
P_{p}\left( (E_{12}^{(4)})^2\right)\approx C_p^{(4)} N,
\ee
so that
$
\Pr[(E_{12}^{(4)})^2\geq C_p^{(4)} N]\approx 1-\frac{p}{100}.
$
In the next subsection, we will check by numerical simulations that this is a good approximation, and also estimate the constant $C^{(4)}$, but from now on assume that this is correct.

Therefore, for sufficiently large $N$, we have that the energy difference $|E_1(\lambda)-E_2(\lambda)|$ becomes larger than 4 with probability $1-\frac{p}{100}$ for 
\be\label{eq:lambda-crossing}
\lambda>\lambda_c=\sqrt{2}\ (C_p^{(4)}N)^{-1/8}.
\ee
Suppose that $E_1(\lambda_*)-E_2(\lambda_*)>4$ for a given $\lambda_*>\lambda_c$ (the situation is depicted in Fig.~\ref{fig:level-crossing}). If we modify the problem by introducing one additional clause, we would still have $E'_1(\lambda_*)>E'_2(\lambda_*)$ since, by definition of $H_P$, one clause can only increase the energy by $4$. However, if this clause is such that it is satisfied by $\x^1$ but not by $\x^2$, we have $E'_1(0)=0$ and $E'_2(0)>0$, meaning that we now have a level crossing between $\lambda=\lambda_*$ and $\lambda=0$.

 If the Hamming distance $n$ between the two solutions scales as the total number of bits $N$, which is typically the case for instances close to the satisfiability threshold~\cite{biroli,zdeborova}, the crossing will only be avoided at the $n$-th order of perturbation theory, so that the minimum gap for the new problem will scale as $\lambda_c^n$, which is exponentially small in $N$. A more detailed study of the scaling of this gap will be provided in Section~\ref{sec:gap-scaling}.

\begin{figure}
\begin{center}
 \includegraphics[width=225pt]{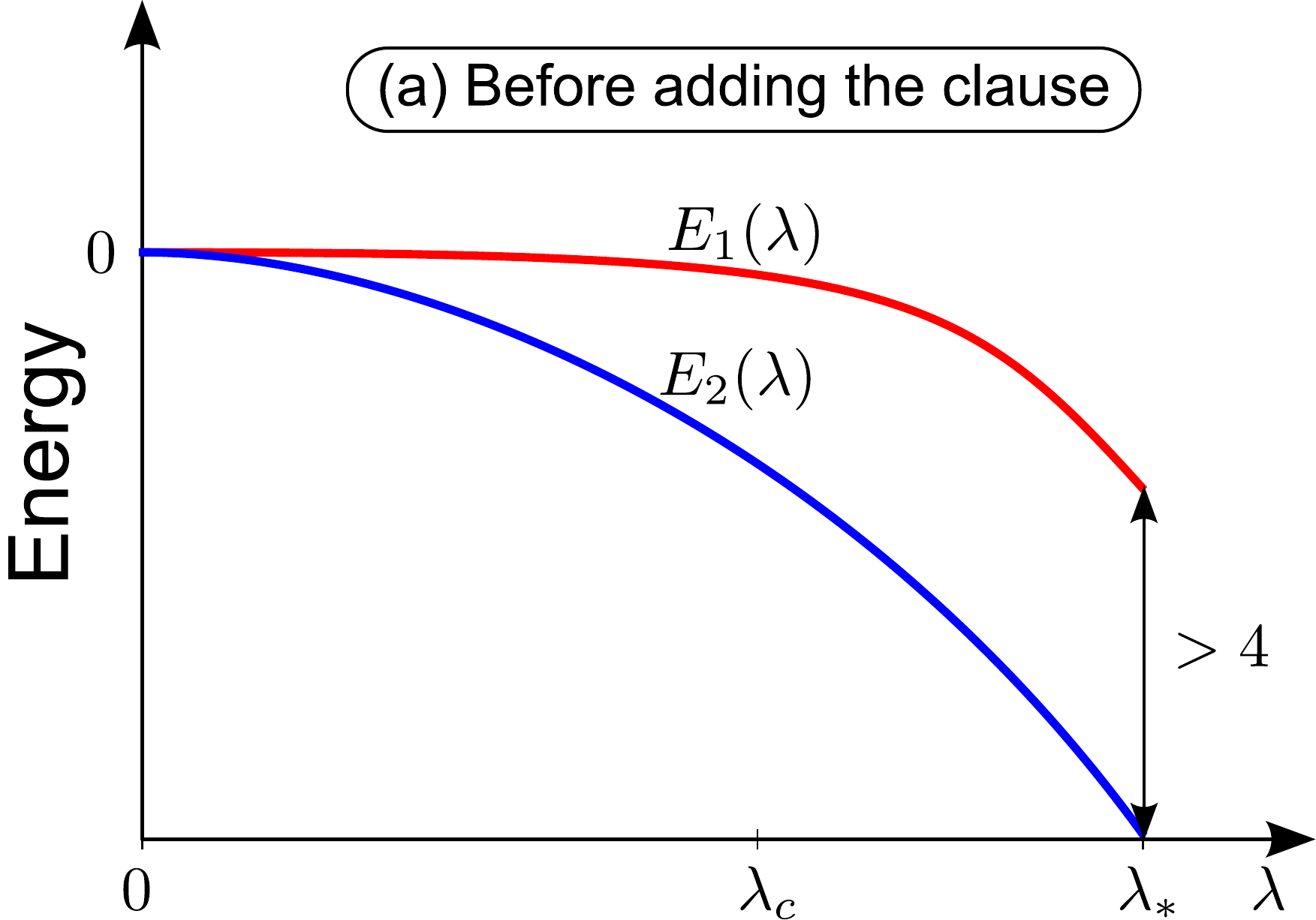}\hspace*{0.5cm}
\includegraphics[width=225pt]{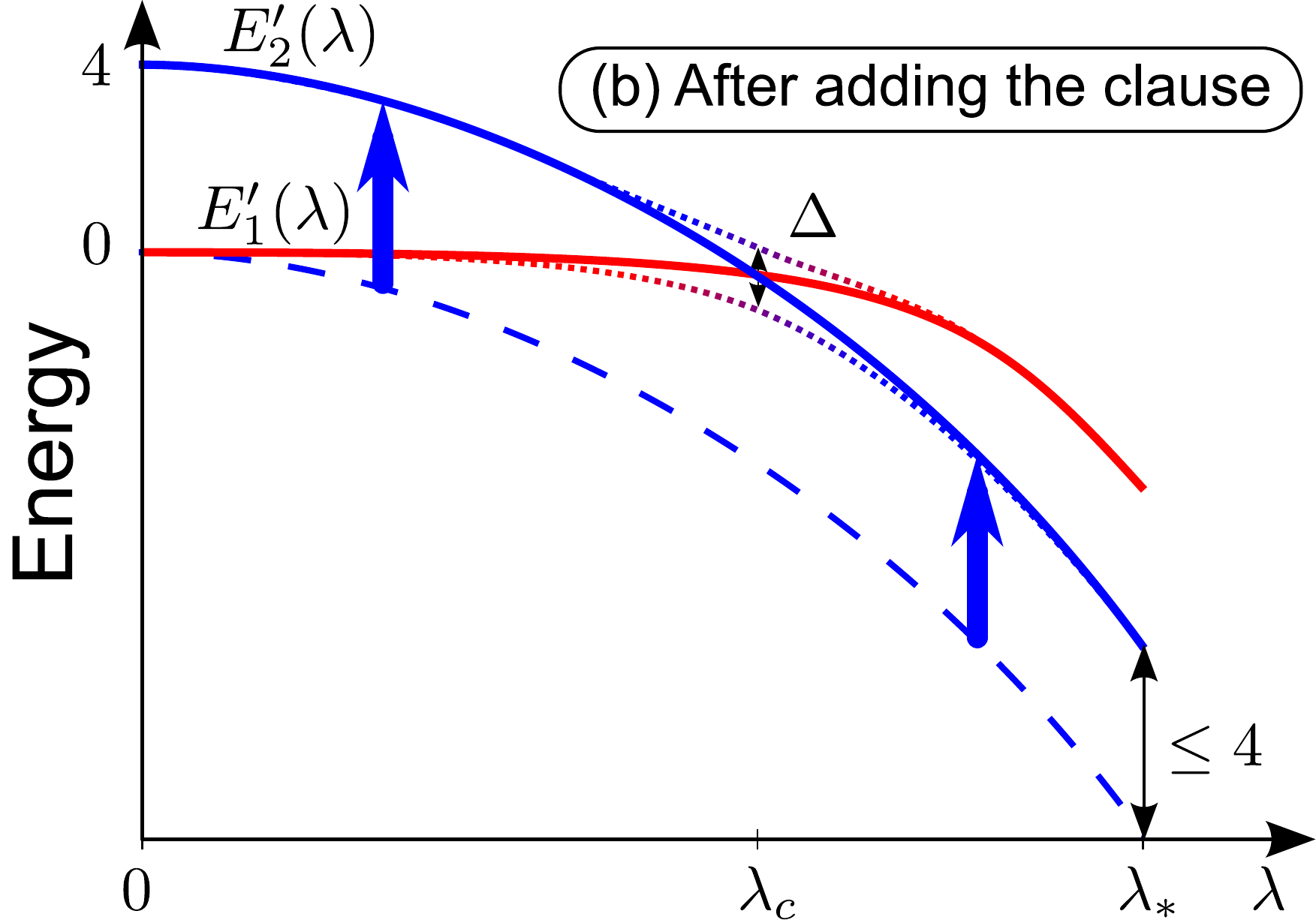}
\end{center}
\caption{Schematical representation of a level crossing. (a) Before adding the clause, we have $E_1(\lambda_*)-E_2(\lambda_*)>4$. (b) By adding a clause satisfied by solution 1 but not solution 2, we create a level crossing since $E_1'(0)<E_2'(0)$ but $E_1'(\lambda_*)>E_2'(\lambda_*)$.}
 \label{fig:level-crossing}
\end{figure}

\subsection{Numerical simulations}
To demonstrate the fact that $(E_{12}^{(4)})^2$ scales as $\Theta(N)$ and obtain an estimation of the slope, we performed the following numerical simulations. Since we are interested in hard instances, accepting very few isolated solutions, we fixed the clauses-to-variables ratio to $\alpha=0.62$, which is close to the satisfiability threshold $\alpha_s$~\cite{raymond07}. For each number of bits $N$  from $15$ to $200$ by steps of 5, we generated 5000 random instances with $M=\lfloor\alpha N\rfloor$ clauses, and then computed the energy splitting $E_{12}^{(4)}$ between the ground state of $H(\lambda)$ and the level that would correspond to the ground state of the Hamiltonian obtained by adding a final clause (the details of this procedure are given in Appendix~\ref{app:numerical}).

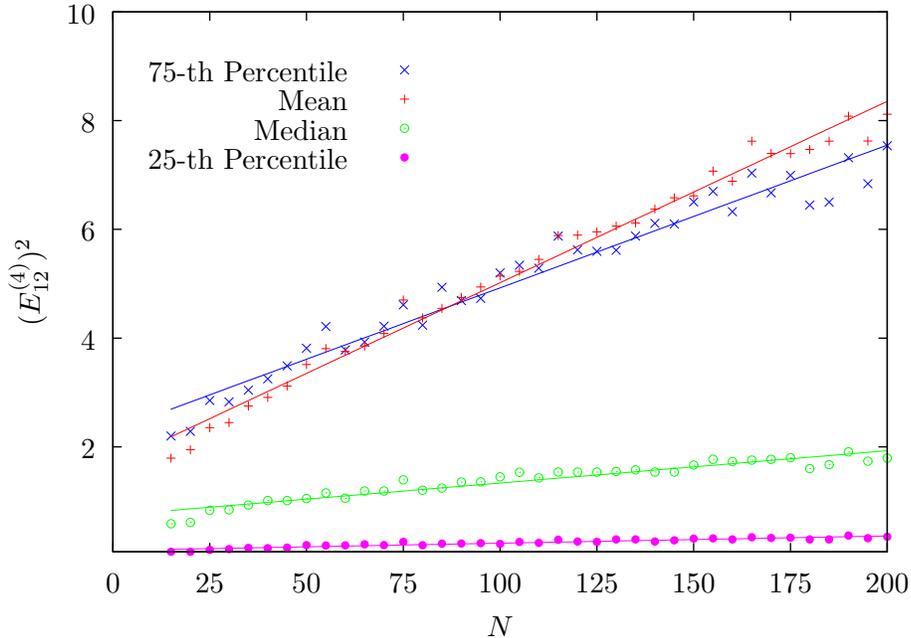
\begin{figure}
 \centering
\input{energy-diff-order4-2local.tex}
 \caption{Statistics of the fourth order correction of the splitting $(E_{12}^{(4)})^2$. Each data point is obtained from 5000 EC3 instances with $\alpha\approx 0.62$.}
 \label{fig:e-value-diff-order4}
\end{figure}

In Fig.~\ref{fig:e-value-diff-order4}, we plotted the mean and some percentiles of $(E_{12}^{(4)})^2$ obtained by our simulation as a function of $N$. The mean closely agree with the linear regression in Eq.~(\ref{linear}),
with $C^{(4)}\approx 0.033$.

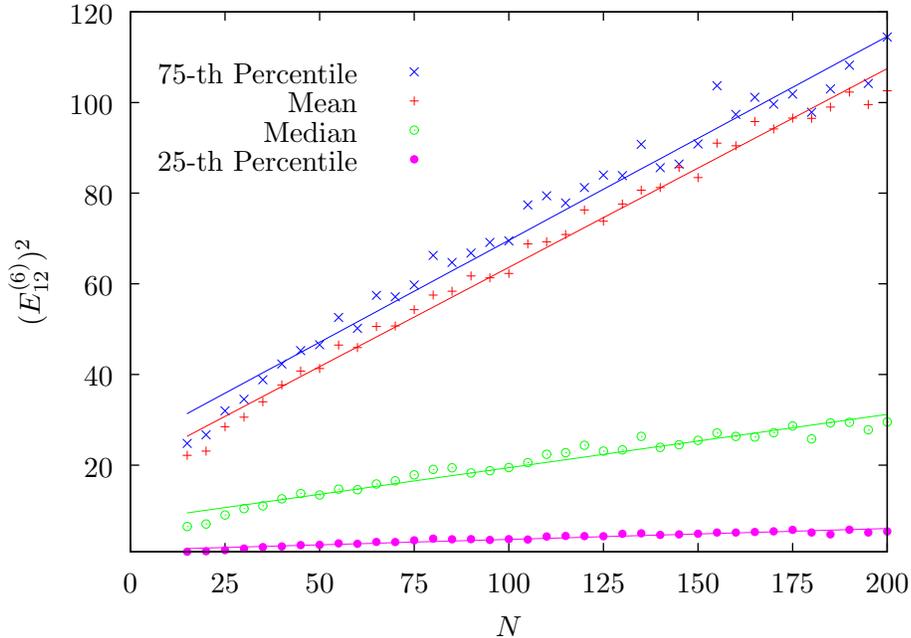
\begin{figure}
 \centering
\input{energy-diff-order6-2local.tex}
 \caption{Statistics of the sixth order correction of the splitting $(E_{12}^{(6)})^2$. Each data point is obtained from 5000 EC3 instances with $\alpha\approx 0.62$.}
 \label{fig:e-value-diff-order6}
\end{figure}

Recall that we have shown in Section~\ref{sec:perturbation} that corrections $E_\x^{(q)}$ up to any order involves $\Theta(N)$ terms. Therefore, just as for order 4, higher order corrections to the splitting squared $(E_{12}^{(q)})^2$ are also expected to scale linearly in $N$, so that the $4$-th order gives the leading behavior of the splitting. To numerically confirm this, we also computed the $6$-th order correction to the splitting for each instance. In Fig~\ref{fig:e-value-diff-order6}, we plotted the mean and some percentiles of $(E_{12}^{(6)})^2$. As expected, they also agree closely with linear regressions, and in particular, we obtain for the mean
$
\mean{(E_{12}^{(6)})^2}\approx C^{(6)} N
$
with $C^{(6)}\approx 0.44$. This also allows us to give a very rough first approximation for the range of $\lambda$ where the perturbation expansion of the splitting
$$
E_{12}(\lambda)=E_{12}^{(4)}\lambda^4+E_{12}^{(6)}\lambda^6+O(\lambda^8),
$$
gives an accurate estimation.
Indeed, for the second term to be less than the first term, we need
$
\lambda<\sqrt{|E_{12}^{(4)}|/|E_{12}^{(6)}|}\approx \lambda_r,
$
and using the linear regression for the means yields $\lambda_r\approx(C^{(4)}/C^{(6)})^{1/4}\approx 0.52$.

\section{Scaling of the gap}\label{sec:gap-scaling}
In this section, we will study the scaling of the gap as the size of the problem increases, and confirm that the gap decreases exponentially. Recall how that gap is created (see Fig.~\ref{fig:level-crossing}), and consider the modified Hamiltonian with the additional clause, which exhibits an avoided crossing for some $\lambda=\lambda_c$. Let us study what happens when we evolve adiabatically from a large $\lambda$ to $\lambda=0$. For $\lambda>\lambda_c$, the ground state corresponds to the energy level $E'_2(\lambda)$, so the system is in the corresponding eigenstate $\ket{\x^2,\lambda}$, while the energy level $E'_1(\lambda)$ corresponds to an excited state $\ket{\x^1,\lambda}$. However, when $\lambda<\lambda_c$, the ground state now corresponds to the energy level $E'_1(\lambda)$, so this means that the system has to tunnel from $\ket{\x^2,\lambda}$ to $\ket{\x^1,\lambda}$. For small enough $\lambda$, these states will be localized, so that the tunneling amplitude will be small if $\x^1$ and $\x^2$ are far apart. More precisely, the minimal gap $\Delta$, that is, the width of the avoided crossing, may be evaluated in the regime of small $\lambda$ by computing the tunneling amplitude
$$
A_{12}(\lambda)=\braket{\x^{1}}{\x^2,\lambda}
$$
between $\ket{\x^2,\lambda}$ and $\ket{\x^1}$, at $\lambda=\lambda_c$. We will show that this amplitude, and therefore the gap $\Delta$ itself, becomes exponentially small if the avoided crossing happens for small enough $\lambda_c$. This implies that unless the evolution is exponentially long, the system will not have the time to tunnel from $\ket{\x^2,\lambda}$ to $\ket{\x^1,\lambda}$, and therefore end up in the state $\ket{\x^2}$, which does not correspond to a solution but only to a local minimum.

\subsection{The Disagree problem}

We will show that computing the tunneling amplitude between two solutions of EC3 reduces to computing the same quantity for an instance of the Disagree problem. An instance of Disagree over $n$ bits consists in a set of clauses of the form $(x_{i_C}\neq x_{j_C})$, where $i_C,j_C\in[n]$. As for EC3, a solution of Disagree is a bit string $\x\in\{0,1\}^n$ satisfying all clauses, and therefore corresponds to the minimum of a cost function
$$
f(\x)=\sum_C (x_{i_C}+x_{j_C}-1)^2.
$$
Therefore, we can design an adiabatic quantum algorithm for Disagree using as final Hamiltonian
$$
H_P=\frac{m}{2}\id+\frac{1}{4}\sum_{i,j=1}^n \tilde{J}_{ij} \sigma_z^{(i)}\sigma_z^{(j)},
$$
where, similarly to EC3, $\tilde{J}_{ij}$ is the number of clauses where bits $i,j$ appear together, and $m$ is the total number of clauses.

Since each clause involves exactly two bits, an instance may be associated to a graph where each vertex $i\in[n]$ represents a bit, and each edge $(i,j)\in[n]^2$ represents a clause $(x_i\neq x_j)$. Note that unless the graph is bipartite, it may include odd cycles and therefore the corresponding problem would have no solution. Here, we will focus on instances associated to connected bipartite graphs, which admit exactly two solutions, where all bits are set to $0$ in one partition and to $1$ in the other. Note that by negating all the bits in one partition, we may map such an instance of Disagree to an instance of Agree, where all clauses are of the type $(x_{i_C}= x_{j_C})$, and where the solutions are the all-$0$ and all-$1$ strings. The Agree problem has been previously studied in the context of adiabatic quantum computing, and it has been shown that when the graph is a cycle, the gap is only polynomially small, but it can be made exponentially small by modifying the weights on the different clauses~\cite{Reich}. Here we show that the Disagree problem is also relevant to the study of EC3.

\subsection{Reduction to the Disagree problem}

\begin{claim}
 Up to leading order in perturbation theory, the tunneling amplitude between two solutions $\x^{1},\x^{2}$ of an instance of EC3 over $N$ bits, is the same as for an instance of Disagree over $n$ bits where the associated graph is bipartite and $n=d_H(\x^{1},\x^{2})$ is the Hamming distance between the solutions.
\end{claim}

\begin{proof}
By perturbation theory, the amplitude of the $\x^{1}$ to $\x^{2}$ transition is given by
\be
A_{12} = \lambda^n \sum_{\y^1,\y^2,\dots ,\y^{n-1}}\frac{V_{\x^{1}\y^1}V_{\y^1\y^2}\dots V_{\y^{n-2}\y^{n-1}}V_{\y^{n-1}\x^{2}}}{E_{\x^{2}\y^1}E_{\x^{2}\y^2}\dots E_{\x^{2}\y^{n-1}}}+O(\lambda^{n+1}),
\ee
where the sum is over bit strings $\y^i\in\{0,1\}^N$. Since $\x^{2}$ is a solution, we have $E_{\x^{2}}=0$ and $E_{\x^{2}\y^i}=-E_{\y^i}$. Moreover, we have $V_{\x\y}=0$ unless $d_H(\x,\y)=1$, so each two successive strings in a path $(\x^{1},\y^1,\y^2,\dots ,\y^{n-1},\x^{2})$ should only differ in one bit, otherwise the corresponding term is zero. Let us assume w.l.o.g. that the bits are labeled so that $x^{1}_i\neq x^{2}_i$ for $1\leq i\leq n$, and $x^{1}_i=x^{2}_i$ otherwise. For a given permutation $p$ in the symmetric group $S_n$, we will denote by $\y^{(p,j)}$ the string obtained from $\x^{1}$ by flipping bits $p(1),p(2),\ldots,p(j)$. We may now write the tunneling amplitude as
\be
|A_{12}| = \lambda^n \sum_{p\in S_n}\frac{1}{E_{\y^{(p,1)}}E_{\y^{(p,2)}}\dots E_{\y^{(p,n-1)}}}+O(\lambda^{n+1}).
\ee
Note that for all the strings in this expression, the bits $y^{(p,j)}_i=x^{1}_i=x^{2}_i$ are constant for $i>n+1$, so these bits are irrelevant and we may consider the restriction of the EC3 instance where these bits are fixed to the same value as for $\x^{1}$ and $\x^{2}$. In that case, all clauses only involving irrelevant bits are trivially satisfied, since they are satisfied for the solutions $\x^{1},\x^{2}$, so these clauses are also irrelevant and may be discarded. Starting from a satisfying assignment for three bits in an EC3 clause, flipping one or three of these bits always makes the clause violated, so relevant clauses may only involve two relevant and one irrelevant bit. Moreover, the relevant bits have to differ in the original assignment to keep the clause satisfied as we flip both of them, so the two relevant bits $i_C,j_C$ in a relevant clause are such that $x^{1}_{i_C}\neq x^{1}_{j_C}$ (similarly for $\x^{2}$). Therefore, discarding irrelevant bits, each relevant clause reduces to a Disagree clause $(x_{i_C}\neq x_{j_C})$. Since the instance of Disagree that we obtain admits two solutions (the restrictions of $\x^{1},\x^{2}$ to the first $n$ bits), the associated graph does not contain any odd cycle and is therefore bipartite.
\end{proof}
Since the graph associated with the Disagree instance is bipartite, we may further reduce it to an Agree instance (by negating all bits in one of the partitions). Therefore, it suffices to study the tunneling amplitude for the Agree problem.

\subsection{Upper bound for the tunneling amplitude\label{sec:upper-bound-amplitude}}
\begin{lemma}
 For any connected graph $G$ on $n$ vertices, the tunneling amplitude between the all-0 and all-1 solutions of the associated Agree instance is $|A_{12}|\leq \frac{1}{2} (2\lambda)^n+O(\lambda^{n+1})$.
\end{lemma}
From there, we can conclude that if an avoided crossing happens at $\lambda=\lambda_c$ the minimal gap will scale as
\be
\Delta=O\left(\left(\frac{\lambda_c}{\lambda_a}\right)^n\right),
\ee
for some constant $\lambda_a>\frac{1}{2}$, so that it becomes exponentially small if $\lambda_c<\frac{1}{2}$.
In Appendix~\ref{Tunneling-amplitude-estimation}, we also provide an estimation of the tunneling amplitude and show that the estimation matches the upper bound up to a constant. In particular, we find $\lambda_a\approx 0.81$ for $\alpha=0.62$.

\begin{proof}
 For a given permutation $p\in S_n$, let $A_p=\prod_{j=1}^{n-1} (E_{\y^{(p,j)}})^{-1}$. We need to prove that $\sum_{p\in S_n}A_p\leq2^{n-1}$. We will actually show that for any tree, $\sum_{p\in S_n}A_p=2^{n-1}$. Since the energies $E_{\y^{(p,j)}}$ can only decrease when removing clauses, the amplitude for a spanning tree of a graph is an upper bound on the amplitude for the graph itself, so this implies the lemma.

We prove that $\sum_{p\in S_n}A_p=2^{n-1}$ for all trees by induction on the size of the tree $n$.
For $n=2$, the only tree consists in two vertices connected by an edge, corresponding to a unique clause $(x_1=x_2)$. There are only two permutations $p\in S_2$, and since flipping either bit will violate the clause, we have $A_p=1$ for each $p$ and finally $A_{12}=2$.

Suppose that $\sum_{p\in S_n}A_p=2^{n-1}$ for all trees of size $n$. We show that this implies that $\sum_{p\in S_{n+1}}A_p=2^{n}$ for all trees of size $n+1$. The set of trees of size $n+1$ may be generated by attaching an additional vertex to any possible vertex of any possible tree of size $n$. Let us now consider a particular tree of size $n$, and suppose we attach an $(n+1)$-th vertex to the vertex number $k$. Let us represent a permutation $p\in S_n$ as a vector $p=(p(1),\ldots,p(n))$. Then all permutations $p'\in S_{n+1}$ may be obtained by considering all permutations $p\in S_n$, and inserting element $(n+1)$ in all possible positions. Let $p_l\in S_{n+1}$ be the permutation obtained from $p\in S_n$ by inserting element $(n+1)$ after element $p(l)$, that is, $p_l=(p(1),\ldots,p(l),n+1,p(l+1),\ldots,p(n))$, where $0\leq l\leq n$ (for $l=0$, the element is inserted in first position, before $p(1)$). We then have $\sum_{p'\in S_{n+1}}A_{p'}=\sum_{p\in S_n}\sum_{l=0}^n A_{p_l}$. We will now show that for any $p\in S_n$, we have $\sum_{l=0}^n A_{p_l}=2A_{p}$, whatever the tree of size $n$ we start with and the vertex $k$ where we attach the additional vertex. This allows to conclude the proof.

Let us fix a permutation $p\in S_n$, and write $E_j=E_{\y^{(p,j)}}$, so that $A_p=\prod_{j=1}^{n-1} (E_j)^{-1}$.
We need to compute $\sum_{l=0}^n A_{p_l}$ for the new Agree instance, obtained by attaching an $(n+1)$-th vertex to the vertex number $k$ of the associated graph, or equivalently introducing a new clause $(x_k=x_{n+1})$. Let $k'=p^{-1}(k)$. Starting from the all-0 solution, let us study how the energy changes as we flip bits in the order specified by $p_l$. For $0\leq l\leq k'-1$ the additional clause will be violated as soon as we flip the new bit $n+1$ (in position $l$), and until we flip bit $k$ (in position $k'$). Therefore, we have
$$
A_{p_l}=\prod_{j=1}^l (E_j)^{-1}\prod_{j'=l}^{k'-1} (E_{j'}+1)^{-1}\prod_{j''=k'}^{n-1} (E_{j''})^{-1},
$$
and it is straightforward to check that $\sum_{l=0}^{k'-1} A_{p_l}=\prod_{j=1}^{n-1} (E_j)^{-1}=A_p$. Similarly, for $k' \leq l \leq n$ we also obtain $\sum_{l=k'}^{n} A_{p_l}=A_p$, and therefore $\sum_{l=0}^{n} A_{p_l}=2A_p$ as claimed.
\end{proof}

\section{Anderson localization and the applicability of perturbation theory\label{sec:applicability}}
As shown in Section~\ref{sec:avoided-crossing}, the spectrum of $H(\lambda)$ presents avoided crossings for $\lambda$ close to zero, and the position of the first avoided crossing $\lambda_c$ is expected to scale as $O(N^{-1/8})$. The presence of an avoided crossing makes the perturbation expansion divergent for $\lambda>\lambda_c$. This means that, strictly speaking, the convergence radius of the perturbation theory scales as $\lambda_c=O(N^{-1/8})$, which tends to zero as $N$ becomes large.

However, this does not imply that we can not use perturbation theory to estimate the spectrum of $H(\lambda)$ beyond $\lambda_c$. For $\lambda>\lambda_c$, the perturbation expansion will be asymptotic rather than convergent. While asymptotic expansions do not generally converge, they may still provide accurate approximations as long as only a finite number of terms are considered. In our case, the divergence is due to an avoided crossing between two levels corresponding to assignments at a Hamming distance $n=\Theta(N)$ from each other. This means that the divergence will only appear at order $n=\Theta(N)$ in perturbation theory, so that successive orders of the perturbation expansion will provide better and better approximations to the spectrum of $H(\lambda)$, as long as we stop at a finite order. On the other hand, for orders higher than $n$, the expansion will start diverging and the approximation will become less and less accurate. The following lemma may be used to obtain an upper bound on the error of the estimation of an eigenvalue.
\begin{lemma}
 Let $\ket{\tilde{\psi}}$ and $\tilde{E}$ be estimations of an eigenstate and corresponding eigenvalue of a Hamiltonian $H$. Then $H$ admits an eigenvalue $E$ such that $|E-\tilde{E}|\leq \norm{(H-\tilde{E}\id)\ket{\tilde{\psi}}}$.
\end{lemma}
\begin{proof}
Let $\epsilon=\norm{(H-\tilde{E}\id)\ket{\tilde{\psi}}}$. By definition, we have $\epsilon^2=\bra{\tilde{\psi}}(H-\tilde{E}\id)^2\ket{\tilde{\psi}}$. This implies that $(H-\tilde{E}\id)^2$ has an eigenvalue $e^2$ such that $0\leq e^2\leq\epsilon^2$. In turn, $H-\tilde{E}\id$ must have an eigenvalue $e$ with $-\epsilon\leq e\leq\epsilon$ or, equivalently, $H$ has an eigenvalue $E=\tilde{E}+e$.
\end{proof}
If the estimations are obtained from $q$-th order perturbation theory, this implies that the error on the eigenvalue may be bounded by computing the $(q+1)$-th order corrections to the eigenvalue and eigenstate. If the $(q+1)$-th order corrections are small compared to the $q$-th order estimation, we still obtain a perfectly valid estimation, even if we are outside of the convergence radius.

Of course, when $\lambda$ becomes much larger, perturbation theory will eventually fail as the error due to higher orders will not be small compared to the estimation anymore. The success of perturbation theory relies on the fact that the perturbed eigenstate is close enough to the unperturbed one. In our case, an unperturbed eigenstate corresponds to a basis state $\ket{\x}$, and the perturbed eigenstate will be close when it has a large overlap on $\ket{\x}$, and a smaller overlap on other basis states $\ket{\y}$, decreasing exponentially with the Hamming distance $d_H(\x,\y)$. Such a state is called a \emph{localized state}. Therefore, the validity of perturbation theory is intimately related to the phenomenon of Anderson localization.

Note that the Hamiltonian $H(\lambda)=H_P+\lambda H_0$ may be written as
\ber
H(\lambda)
&=&\sum_{\x\in\{0,1\}^N}E_\x\proj{\x}-\lambda\sum_{\x,\y:d_H(\x,\y)=1}\ket{\x}\bra{\y},
\eer
where the second term (kinetic energy) describes a particle hopping on the vertices of a hypercube, and the first term (potential) adds random energies to each vertex of the hypercube (here the randomness comes from the distribution of random instances).
This expression is therefore very similar to the Hamiltonian used by Anderson to demonstrate his famous localization effect in disordered systems~\cite{anderson58}, the main difference being that the geometry of the system is the $N$-dimensional hypercube instead of the $d$-dimensional lattice. The key feature of Anderson localization is that for large enough disorder, or in our terms for small enough $\lambda$, the eigenstates of $H(\lambda)$ are localized (i.e., the particle is bound to a vertex of the hypercube) so that perturbation theory is applicable in this regime.
On the other hand, for large $\lambda$, the first term (corresponding to the disorder) is negligible and the eigenstates of $H(\lambda)$ are close to the eigenstates of $H_0$, describing waves travelling through the hypercube.
Therefore, for large $\lambda$, the eigenstates of $H(\lambda)$ will have amplitude over all basis states $\ket{\x}$ and be called \emph{extended}.

The transition from localized to extended states marks the failure of perturbation theory, and the position $\lambda_r$ of this transition may be viewed as a weak notion of convergence radius, as shown by Anderson. To circumvent the divergence of perturbation theory due to avoided crossings close to $\lambda=0$, Anderson proposed the following resolution. He added a small imaginary part $i\eta$ to all energies $E_\x$ of the unperturbed Hamiltonian. After this, perturbation theory provides a convergent series, which Anderson called locator expansion. Then, he took the limit $\eta\to 0$, but only after taking the limit $N\to\infty$. Recall that the original divergence of the perturbation expansion started at order $n=\Theta(N)$. By taking $N\to\infty$ before $\eta\to 0$, this keeps the series convergent and the radius of convergence $\lambda_r$ then corresponds to the transition from localized to extended states (note that if we take the limits in the other order, we would obtain the original radius of convergence $\lambda_c$ due to avoided crossings, which goes to zero as $O(N^{-1/8})$).

For conventional Anderson localization, the contribution $i\eta$ to the energies can be due to a weak coupling to a thermal bath, or to a possible escape of the particle through the boundary of the system. In the case of adiabatic quantum optimization, the basic idea is that $i\hbar/T$ (where $T$ is the total computation time) would play the same role as $i\eta$. Indeed let $\Delta$ be the eigenvalue gap created by an avoided crossing. During an evolution of time $T$, the system cannot resolve energies smaller than $\hbar/T$. If $\Delta\ll\hbar/T$, the system will therefore behave as if the avoided crossing was an actual crossing between two decoupled levels, and the divergence of the perturbation expansion due to this avoided crossing will disappear. We have shown in the previous section that $\Delta$ is exponentially small, which implies that this argument is consistent unless we consider an exponentially long time $T>\hbar/\Delta$. While a complete analysis of Anderson localization in adiabatic quantum optimization could be the subject of an article on its own, the only property that we need is that the position of the transition $\lambda_r$ decreases slower than the position of the first avoided crossing $\lambda_c$. For Anderson localization on a $d$-dimensional lattice with random energies of order $E$, 
$\lambda_r$ is believed to scale as $\Theta({E}/{(d\ln d)})$~\cite{efetov,abou-chacra}. In our case, both the degree of vertices and the random energies scale as $\Theta(N)$, so we can expect by analogy that $\lambda_r=\Omega(1/\ln N)$, which is sufficient for our result since $\lambda_c=O(N^{-1/8})$. This is probably pessimistic since the conventional Anderson localization considers the critical $\lambda$ for which \emph{all} eigenstates are still localized, while we are only interested in low energy eigenstates which remain localized much longer (as the density of eigenvalues is much larger in the middle of the spectrum, this is where the first extended states appear). Moreover, since we have shown in Section~\ref{sec:perturbation} that all terms in the perturbation expansion scale as $\Theta(N)$, meaning that dividing the expansion by $N$ gives an expression which is essentially $N$-independent, it is actually likely that $\lambda_r=\Theta(1)$ for low energy states.

\section{Conclusion and discussion}
Putting everything together, we conclude that the adiabatic quantum optimization algorithm fails for random instances of EC3 because of the presence of an exponentially small gap close to the end of the evolution (at $s=1$ or $\lambda=0$). This result relies on different elements.
In Section~\ref{sec:perturbation}, we have shown that all corrections in the perturbation expansion of eigenenergies of $H(\lambda)$ scale as $\Theta(N)$. In Section~\ref{sec:avoided-crossing}, we have shown that perturbation theory predicts an avoided crossing to occur for some $\lambda_c=O(N^{-1/8})$, and confirmed this scaling by numerical simulations.
In Section~\ref{sec:applicability}, we have also argued that by analogy to the usual Anderson model, we can expect states to be localized for small enough $\lambda<\lambda_r=\Omega(1/\log N)$, so that perturbation theory is applicable in this regime.
Since, as $N$ increases, $\lambda_c$ decreases faster than $\lambda_r$, for large enough $N$ the expected avoided crossing occurs for $\lambda_c<\lambda_r$, where the perturbation expansion is accurate, so that the prediction of perturbation theory becomes valid. Finally, in Section~\ref{sec:gap-scaling}, we have shown that the gap induced by an avoided crossing at $\lambda_c$ will scale as $\Delta=O(\lambda_c/\lambda_a)^n$, where $\lambda_a$ is a constant and $n$ is the Hamming distance between the assignments corresponding to the two levels. For large instances close to the satisfiability threshold $\alpha_s$, random instances will only have a few isolated solutions which are essentially independent from each other, so that $n=\Theta(N)$. Moreover, since $\lambda_c=O(N^{-1/8})$, this implies that the gap will scale as $\Delta=2^{-O(N\log N)}$, and therefore the algorithm will fail unless it takes (super-)exponential time. It is interesting to note that this scaling matches a lower bound proved by van Dam and Vazirani (see~\cite{vdam2}, the proof is given in Appendix~\ref{sec:lower-bound} for completeness), so that the typical gap for random instances is really as small as it can get. The fact that the gap is actually super-exponentially small also implies that the complexity of the AQO algorithm will be even worse than classical solvers (including the naive algorithm that exhaustively checks all possible assignments).

While our results predict an exponentially small gap to occur closer and closer to the end of the evolution (at $s=1$), an obvious question is why this was not observed in previous numerical simulations of the gap, where it seemed to scale polynomially and occur close to a fixed $s$, at least for small $N$ (up to $N\lesssim 20$ in~\cite{fggllp01,hogg03}, 60 in~\cite{bolpr06} and 128 in~\cite{yks08}, each using different methods). As for the position of the gap, note that it scales as $\lambda_c=O(N^{-1/8})$, so the dependence on $N$ is weak and we need to consider a very broad range of $N$ to observe it. More importantly, $N$ must also be large enough so that the avoided crossing happens in a region where perturbation theory is applicable, i.e., $\lambda_c<\lambda_r$, and where the tunneling amplitude decreases exponentially, i.e., $\lambda_c<\lambda_a$. Since we have found the condition $\lambda_c>\sqrt{2}\ (C^{(4)}N)^{-1/8}$ for an avoided crossing to occur at $\lambda_c$, and obtained the estimations $\lambda_r\approx 0.52$ and $\lambda_a\approx 0.81$, this implies the bound $N>\frac{16}{C^{(4)}\lambda_r^8}\approx 86000$ bits to observe the exponentially small gap predicted by our approach. While this number is very high, practically the exponentially small gap could appear much faster. A first reason is that in many cases, an additional clause will only increase the energy by 1, and not the maximum 4, due to the form of the cost function in Eq.~(\ref{eq:cost-function}). If we only impose that $E_{12}(\lambda)>1$, we see that small gaps already occur with high probability for $N>\frac{1}{C^{(4)}\lambda_r^8}\approx 5400$. A second reason is that in our numerical simulations, we only considered avoided crossings created by the addition of the last clause, but any of the $M-1$ other clauses could induce other avoided crossings, so that it could be possible to observe this exponentially small gap as soon as $N$ is of the order of a few hundred bits. In particular, let us note that very recent numerical simulations up to $N=256$ by Young~\textit{et al.} have revealed that more and more instances exhibit a first order phase transition and therefore a very small gap~\cite{yks09}, which could be due to the mechanism described in the present paper. In Figure~\ref{fig:simulated-crossing-final}, we plotted a level crossing predicted by fourth order perturbation theory for an instance with $N=200$ obtained during our numerical simulations. The crossing occurs at $\lambda\approx 0.51$ (possibly inside the region where perturbation theory is valid), and the corresponding assignments are at distance $n=60$ from each other.

\begin{figure}
\begin{center}
\includegraphics[width=250pt]{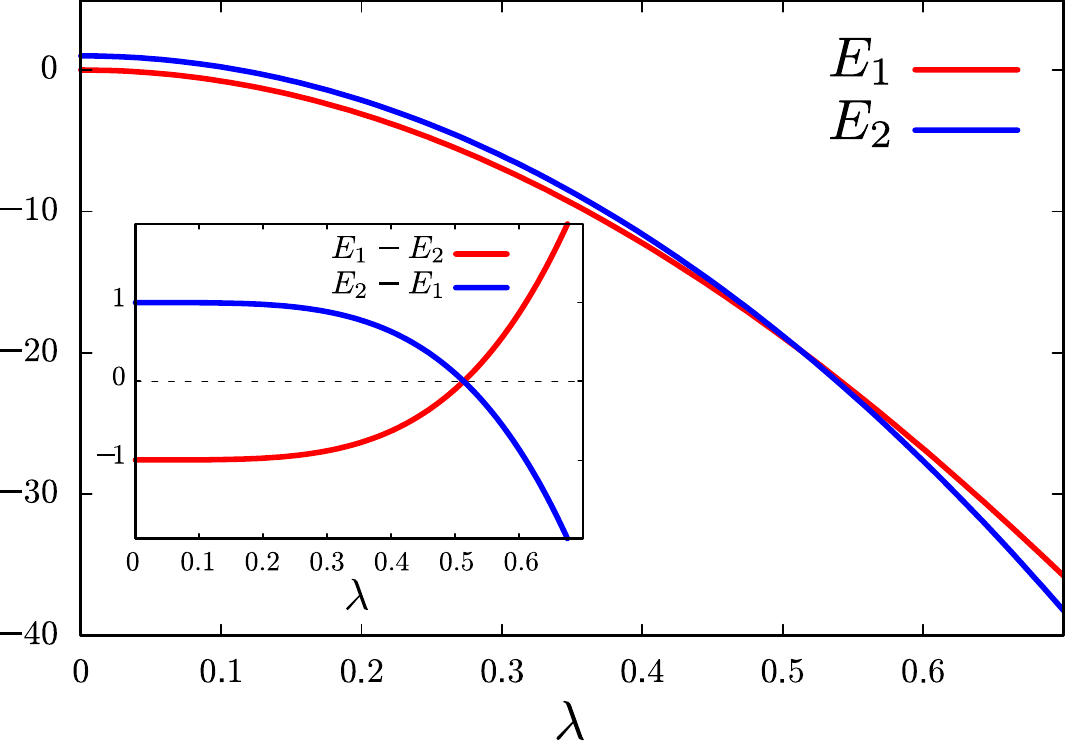}\hspace*{0.5cm}
\end{center}
 \caption{Simulation of a level crossing for a random instance with $N=200$ bits and $\alpha\approx 0.62$, obtained by fourth order perturbation theory. Inset: To make the crossing more apparent, we plotted the energy differences $E_1-E_2$ and $E_2-E_1$.\label{fig:simulated-crossing-final}}
\end{figure}

While our argument immediately implies that a particular quantum adiabatic algorithm for EC3 will fail, it is important to note that it also applies to more general cases. In particular, it does not rely on the specific form of the problem Hamiltonian $H_P$, but rather on its general statistical properties, so that it should extend to other NP-complete problems such as 3-SAT. Moreover, our argument does not rely on the precise form of the initial Hamiltonian $H_0$ either, or even on the possible path between the initial and the final Hamiltonian, but only on the behavior of the perturbation in the vicinity of the final Hamiltonian $H_P$, and we would actually obtain similar conclusions for any perturbation acting locally on the qubits. Therefore, any adiabatic quantum algorithm aimed at solving a similar NP-complete optimization problem by mapping its cost function to a final Hamiltonian would eventually fail for the same reason. Actually, for other problems the situation could be even worse, in that the small gap could occur for smaller instances than for EC3. Indeed, while for EC3 the first non-zero correction to the splitting $E_{12}(\lambda)$ between two solutions only appears at order 4 of perturbation theory, for many other problems, such as 3-SAT, a similar Hamiltonian would already exhibit a splitting at order 2, so that the position of the avoided crossing would scale as $\lambda_c=O(N^{-1/4})$, and therefore this avoided crossing would happen inside the region where perturbation theory is valid for even smaller instances.

\begin{acknowledgments}
 We thank A. Childs, E. Farhi, J. Goldstone, S. Gutmann, M. R\"otteler and A. P. Young for interesting discussions. We thank the \emph{High Availability Grid Storage} department of NEC Laboratories America for giving us access to their cloud to run our numerical simulations. This research was supported in part by US DOE contract No. AC0206CH11357.
\end{acknowledgments}

\bibliography{adiabatic-full-paper}

\appendix

\section{Proof of Lemma~\ref{lem:mean-size-connected}\label{app:number-graphs}}
The proof of the lemma relies on the following claim.
\begin{claim}\label{claim:proba-connected}
 For any $u\in\mathbb{N}$, there exists $g_u^L,g_u^U\in\mathbb{R}$ and $N_u\in\mathbb{N}$ such that $\frac{g_u^L N}{\binom{N}{u}}\leq\Pr[\textrm{``$[u]$ is connected''}]\leq\frac{g_u^U N}{\binom{N}{u}}$ for any $N\geq N_u$.
\end{claim}
We first show that this implies Lemma~\ref{lem:mean-size-connected}, and then prove Claim~\ref{claim:proba-connected}.

\begin{proof}[Proof of Lemma~\ref{lem:mean-size-connected}]
 By definition, we have
$
G_u=\sum_{S\subseteq[N]:|S|=u} I[\textrm{``$S$ is connected''}],
$
where $I[\mathcal{E}]$ is the indicator variable of event $\mathcal{E}$.
Since all subsets of size $u$ are equivalent up to some permutation of the bits, we have
$$
\mean{G_u}=\sum_{S\subseteq[N]:|S|=u} \Pr[\textrm{``$S$ is connected''}]
=\binom{N}{u}\Pr[\textrm{``$[u]$ is connected''}],
$$
which, together with Claim~\ref{claim:proba-connected}, implies that $\mean{G_u}=\Theta(N)$.
\end{proof}

\begin{proof}[Proof of Claim~\ref{claim:proba-connected}]
Let us start with the lower bound. For subsets $S_1,S_2,S_3\subseteq[N]$, we say that a clause $(x_{i_1},x_{i_2},x_{i_3})$ is of type $(S_1,S_2,S_3)$ if $i_k\in S_k$, for $k=1,2,3$. Let
$$
p_u=\frac{u}{N(N-1)}
$$
be the probability that a given random clause is of type $([u],\{u+1\},[N])$, and let us denote by $Conn(S)$ the event ``$S$ is connected''. A sufficient condition for $[u+1]$ to be connected is that for each $q\in[u]$, there is exactly one clause of type $([q],\{q+1\},[N])$. Therefore, we have
\ber
\Pr[Conn([u+1])]&\geq& \frac{M!}{(M-u)!}(\prod_{q=1}^u p_q)(1-\sum_{q=1}^u p_q)^{M-u}\\
&\geq& \frac{M!u!}{(M-u)!N^u(N-1)^u}\left(1-\frac{(M-u)u(u+1)}{2(N-1)}\right)\\
&\geq& \frac{C_{u+1}N}{\binom{N}{u+1}},
\eer
where we have assumed that $N\geq \alpha u(u+1)+1$ so that the second factor is  larger than $1/2$, and defined $C_{u+1}$ as
$$
C_{u+1}=\binom{N}{u+1}\frac{M!u!}{2(M-u)!N^{u+1}(N-1)^u}.
$$
Note that $C_2=\alpha/4$ and
$$
C_{u+1}=\frac{C_u u(N-u)(M-u+1)}{(u+1)N(N-1)}\geq \frac{C_u\alpha}{8},
$$
for $N\geq 2u-1$ and $M\geq 2u-2$.
Therefore, the lower bound holds with $g_2^L=\alpha/4$, $g_{u+1}^L={g_u^L\alpha}/{8}$ and $N$ sufficiently large.

As for the upper bound, we prove it by induction on $u$. For $u=1$, we have $\Pr[\textrm{``$[1]$ is connected''}]\leq 1$ so the claim holds with $g_1^U=1$.
Let us denote by ``$Cl(S_1,S_2,S_3)$'' the event that there exists at least one clause of type $(S_1,S_2,S_3)$ within the $M$ clauses. Now assume that $\Pr[Conn([u])]\leq \frac{g_u^U N}{\binom{N}{u}}$ for some constant $g_u^U$. We need to bound the probability that $[u+1]$ is connected. If $[u+1]$ is connected, there exists some vertex $i\in[u+1]$ such that the graph obtained from $[u+1]$ by removing vertex $i$, i.e., $[u+1]\setminus\{i\}$, is connected as well. Therefore
\ber
&&\Pr[Conn([u+1])] \nonumber \\
&&\leq 6\sum_{i\in[u+1]}\Pr[Conn([u+1]\setminus \{i\}) \wedge Cl(\{i\},[u+1]\setminus \{i\},[N])]\nonumber\\
&&= 6(u+1) \Pr[Conn([u]) \wedge Cl(\{u+1\},[u],[N])]\nonumber\\
&&\leq 6(u+1)\Pr[Conn([u]) \wedge Cl(\{u+1\},[u],[N]\setminus[u])]\nonumber\\
&&
+6(u+1)\Pr[Conn([u]) \wedge Cl(\{u+1\},[u],[u])].\label{eq:Conn_u}
\eer
For the first term, we have
\ber
&&\Pr[Conn([u]) \wedge Cl(\{u+1\},[u],[N]\setminus[u])] \nonumber \\
&&=\Pr[Cl(\{u+1\},[u],[N]\setminus[u])]
\cdot\Pr[Conn([u])\ |\ Cl(\{u+1\},[u],[N]\setminus[u])]\nonumber\\
&&\leq\Pr[Cl(\{u+1\},[u],[N]\setminus[u])]
\cdot\Pr[Conn([u])]\nonumber\\
&&\leq M\frac{u(N-u-1)}{N(N-1)(N-2)}
\cdot\Pr[Conn([u])]\nonumber\\
&&\leq \frac{\alpha u}{N-1}
\cdot\Pr[Conn([u])].
\eer
Similarly, for the second term of Eq.~(\ref{eq:Conn_u}), we have
\ber
&&\Pr[Conn([u]) \wedge Cl(\{u+1\},[u],[u])] \nonumber \\
&&=\Pr[Cl(\{u+1\},[u],[u])]
\cdot\Pr[Conn([u])\ | \ Cl(\{u+1\},[u],[u])]\nonumber\\
&&\leq M\frac{u(u-1)}{N(N-1)(N-2)}\Pr[Conn([u])\ | \ Cl([N],[u],[u])]\nonumber\\
&&\leq \frac{\alpha u(u-1)}{(N-1)(N-2)}\frac{\Pr[Cl([N],[u],[u])\ | \ Conn([u])]\cdot\Pr[Conn([u])]}{\Pr[Cl([N],[u],[u])]}\nonumber\\
&&= \frac{\alpha u(u-1)}{(N-1)(N-2)}\frac{\Pr[Conn([u])]}{\Pr[Cl([N],[u],[u])]}.
\eer
We need to upper bound $\Pr[Cl([N],[u],[u])]$. Let
$$
q_u=\frac{u(u-1)}{N(N-1)}
$$
be the probability that a given random clause is of type $([N],[u],[u])$.
For $N\geq 2\alpha u(u-1)+1$, we obtain
$$
\Pr[Cl([N],[u],[u])]\geq Mq_u(1-Mq_u)\geq \frac{\alpha u(u-1)}{2(N-1)}\geq \frac{\alpha u(u-1)}{4(N-2)}.
$$
Therefore, the second term in Eq.~(\ref{eq:Conn_u}) satisfies
$$
\Pr[Conn([u]) \wedge Cl(\{u+1\},[u],[u])]
\leq\frac{4}{N-1}
\cdot\Pr[Conn([u])].
$$

Putting everything together, Eq.~(\ref{eq:Conn_u}) implies
\ber
\Pr[Conn([u+1])]&\leq&
\frac{6(u+1)}{N-1}(\alpha u+4)\Pr[Conn([u])]\nonumber\\
&\leq&
\frac{6(u+1)}{N-1}(\alpha u+4)\frac{g_u^U N}{\binom{N}{u}}\nonumber\\
&\leq&6g_u^U(\alpha u+4)\frac{N}{\binom{N}{u+1}},
\eer
so that the upper bound holds with $g_1^U=1$, $g_{u+1}^U=6g_u^U(\alpha u+4)$, and $N$ sufficiently large.
\end{proof}

\section{Details of the numerical simulations\label{app:numerical}}
For each number of bits $N$ from $15$ to $200$ by steps of 5, we generated 5000 random instances with $M=\lfloor\alpha N\rfloor$ clauses using the following procedure.

We first picked $M-1$ triplets of bits (the clauses) uniformly at random and with replacement, and cleaned the corresponding EC3 instance by removing absent bits, and clauses involving two bits appearing in no other clause, which introduce trivial degeneracies in the space of solutions. We then solved the cleaned EC3 instance $P$ using the Davis-Putnam-Logemann-Loveland algorithm~\cite{dp60,dll62}, to find all of its satisfying assignments (note that this step takes exponential time and therefore prevents us from considering too large instances). We added a final random clause to create a new instance $P'$, and checked which one of the solutions of $P$ satisfied the new clause, therefore being also solutions of $P'$. For each of the solutions of $P'$, we computed the fourth order correction to the energy $E_\x^{(4)}$ (under Hamiltonian $H'(\lambda)$), and identified the solution $\x^1$ having the minimum correction, which corresponds to the ground state of $H'(\lambda)$ for $\lambda>0$. We then computed the correction $E_\x^{(4)}$ under $H(\lambda)$, not only for $\x^1$ but also for all solutions of $P$ not being solutions of $P'$, in order to identify the solution $\x^2$ having the minimum correction.
Finally, we calculated the 4-th order correction $E_{12}^{(4)}$ of the splitting between energies of assignments $\x^1$ and $\x^2$ for the instance without the added clause $P$. Indeed, if this splitting is large enough, this creates an avoided crossing between the ground state $\ket{\x^1,\lambda}$ and the excited state $\ket{\x^2,\lambda}$ of the Hamiltonian $H'(\lambda)$ for the instance with the added clause.

Note that if any of the previous steps failed for some reason (e.g. $P$ has no solution, or all solutions of $P$ are also solutions of $P'$), we rejected the instance and started over with $M-1$ new random clauses. Nevertheless, after having generated 5000 valid instances, we checked that the number of rejected instances stayed approximately constant as $N$ increases, so that our conclusions are verified with constant probability over the random instances.

The numerical simulations were run using a parallel algorithm on 2 nodes of the NECLA cloud, where each node has access to two dual-core Intel Xeon 5160 (3GHz) CPUs and 3.5GB of RAM. The full set of simulations took between 3 and 4 days to complete.

\section{Estimation of the tunneling amplitude}\label{Tunneling-amplitude-estimation}

The case of a tree provides an upper bound on the tunneling amplitude between two solutions of an Agree instance over $n$ bits, and as a consequence an upper bound on the tunneling amplitude between two solutions of an EC3 instance with Hamming distance $n$. In this appendix, we will derive an estimation (rather than an upper bound) of this tunneling amplitude. Let us consider an arbitrary Agree instance over $n$ bits with $m$ clauses. Let $\tilde{B}_i$ be the number of clauses where bit $i$ appears, and $\tilde{J}_{ij}$ be the number of clauses where bit $i$ and $j$ appear together (there is actually only one such clause for the Agree problem, $(x_i=x_j)$, but in general this clause could be repeated). Following the notations of Section~\ref{sec:upper-bound-amplitude}, let $E_j=E_{\y^{(p,j)}}$ be the energy of the bit string $\y^{(p,j)}$, obtained from $\x^{1}$ by flipping bits $p(1),p(2),\ldots,p(j)$ for a given permutation $p\in S_n$, and let $A_p=\prod_{j=1}^{n-1} (E_j)^{-1}$. We need to estimate
\ber
|A_{12}| &=&\lambda^n\sum_{p\in S_n} A_p+O(\lambda^{n+1})\nonumber\\
&=&\lambda^n n! \mean{A_p}_p+O(\lambda^{n+1}),
\eer
where the average is taken over permutations $p\in S_n$. Let us consider the approximation
\ber
 \mean{A_p}_p&=&\mean{\prod_{j=1}^{n-1} (E_j)^{-1}}_p\nonumber\\
&\approx&\prod_{j=1}^{n-1} \mean{E_j}_p^{-1}.
\eer
We need to compute the average energy $E_j$ of a bit string obtained by flipping $j$ bits from a solution. When a particular bit is flipped, each clause where the bit appears will become violated, but it will be satisfied again when the second bit in the clause is flipped, so that
\be\label{eq:energy-agree}
E_j=\sum_{k=1}^j \left( \tilde{B}_{p(k)}-\sum_{l=1}^j \tilde{J}_{p(k)p(l)}\right).
\ee
 Since, by definition, $\sum_{i=1}^n \tilde{B}_i=2m$ and $\sum_{j=1}^n \tilde{J}_{ij}=\tilde{B}_i$, we have $\mean{\tilde{B}_{p(k)}}_p=\frac{2m}{n}$ and $\mean{\tilde{J}_{p(k)p(l)}}_p=\frac{2m}{n(n-1)}$ (for any $k,l\in[n]$), and inserting in Eq.~(\ref{eq:energy-agree}) yields
$$
\mean{E_j}_p=\frac{2mj(n-j)}{n(n-1)}.
$$
This shows that the average energy barrier that has to be broken to tunnel from one solution of Agree (and by reduction of EC3) to another is shaped as a parabola, see Fig.~\ref{fig:energy-barrier}. We may now estimate the tunneling amplitude
\ber
 |A_{12}|&\approx&\lambda^n n!\prod_{j=1}^{n-1} \mean{E_j}_p^{-1}\nonumber\\
&\approx& \frac{\lambda^n n}{(n-1)!}\left(\frac{n-1}{2\beta}\right)^{n-1},
\eer
where we have defined the clauses-to-variables ratio $\beta=\frac{m}{n}$ for the Agree instance. Using the Stirling formula $n!\approx n^n e^{-n}$, we obtain for large $n$
$$
|A_{12}|\approx\frac{2\beta n}{e}\left(\frac{e\lambda}{2\beta}\right)^n,
$$
which is therefore exponentially small in $n$ as soon as $\lambda<\lambda_a\approx\frac{2\beta}{e}$.

\begin{figure}
 \centering
\includegraphics[scale=1]{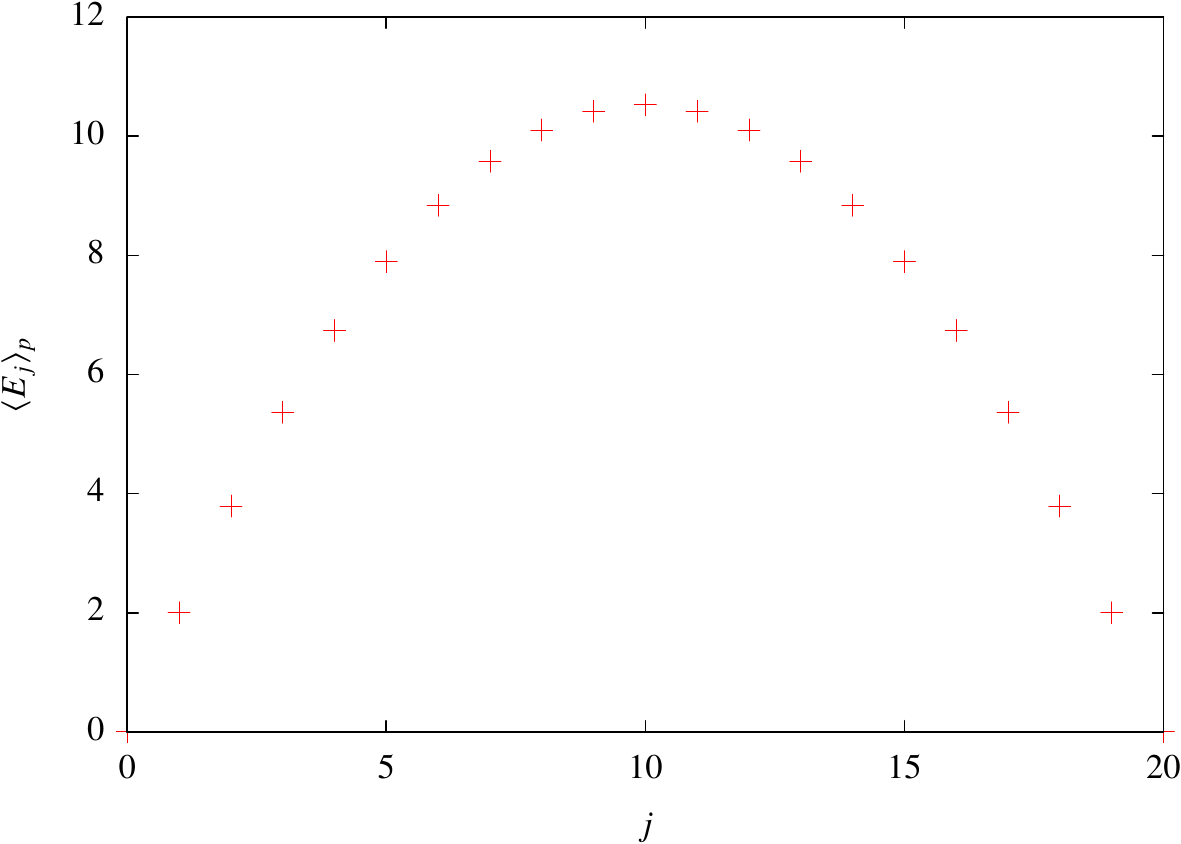}
 \caption{Average energy barrier for random paths between the two solutions of an instance of Agree with $n=20$ bits and $m=20$ clauses. $\mean{E_j}_p$ gives the average energy after flipping $j$ bits from one of the solutions.}
 \label{fig:energy-barrier}
\end{figure}

Since we are more interested in EC3 than Agree, we need to estimate the typical value of $\beta$ for the Agree instances obtained by reduction from EC3 instances. Let us consider random instances of EC3 with $N$ bits and $M=\alpha N$ clauses. We have seen in Section~\ref{sec:statistics} that the clauses forming such a random instance will typically only involve about $\typ{N}'=N(1-e^{-3\alpha})$ of the $N$ bits. Let us discard the absent bits and focus on assignments of the remaining bits satisfying all clauses, that is, the solutions of the instance. Since, for each clause of $3$ bits, one of the bits has to be set to 1 and the other two to 0, we expect that solutions will typically involve about $\frac{\typ{N}'}{3}$ 1's and $\frac{2\typ{N}'}{3}$ 0's. Therefore, if we look at a random bit $x_i$ in a solution, we would expect that $\Pr_i[x_i=1]=\frac{1}{3}$. Let us now consider the joint distribution $p_{b^1b^2}=\Pr_i[(x^{1}_i,x^{2}_i)=(b^1,b^2)]$ for two solutions $\x^{1},\x^{2}$.
 For instances with large $N$ and $\alpha$ close to the satisfiablity threshold, the solutions will be essentially independent, so we can expect $p_{b^1b^2}=\Pr_i[x^{1}_i=b^1]\cdot\Pr_i[x^{2}_i=b^2]$, that is, $p_{00}=\frac{4}{9}$, $p_{01}=p_{10}=\frac{2}{9}$ and $p_{11}=\frac{1}{9}$. From there, we can estimate the typical Hamming distance between two solutions of a random EC3 instance, $\typ{n}=\typ{N}'(p_{01}+p_{10})=\frac{4\typ{N}'}{9}$.

Out of the $M$ clauses, we now need to estimate the number of relevant clauses, that is those involving bits that differ in the two solutions, since these are the clauses corresponding to agree clauses in the Agree instance obtained by reduction. Let us consider the set of all clauses satisfied by the two solutions, and assume that we sample $M$ clauses uniformly at random from this set. We say that a bit $x_i$ is of type $b^1b^2$ when $(x^{1}_i,x^{2}_i)=(b^1,b^2)$. A relevant clause involves one $00$-type bit, one $01$-type bit and one $10$-type bit. Therefore, there are typically about $M_{\mathrm{rel}}=6 p_{00}p_{01}p_{10}(\typ{N}')^3$ possible relevant clauses. Similarly, an irrelevant clause involves one $11$-type bit and two $00$-type bits, so there are typically about $M_{\mathrm{irr}}=3p_{11}p_{00}^2(\typ{N}')^3$ irrelevant clauses. A random clause satisfied by the two solutions will therefore be relevant with probability $\frac{M_{\mathrm{rel}}}{M_{\mathrm{rel}}+M_{\mathrm{irr}}}=\frac{2}{3}$. As a consequence, when picking $M$ random clauses, there will typically be around $\typ{m}=\frac{2M}{3}$ relevant clauses. Finally, the typical clauses-to-variables ratio for the Agree instance obtained by reduction may be estimated as $\typ{\beta}=\frac{\typ{m}}{\typ{n}}=\frac{3}{2}\alpha\frac{e^{3\alpha}}{e^{3\alpha}-1}$. Back to our original concern about the tunneling amplitude, we see that it will be exponentially small for $\lambda<\lambda_a$, where
$$
\lambda_a\approx\frac{2\typ{\beta}}{e}=\frac{3\alpha e^{3\alpha-1}}{e^{3\alpha}-1}.
$$
For $\alpha=0.62$ (the value we used in our numerical simulations), we obtain $\lambda_a\approx 0.81$.

\section{Lower bound on the minimum gap}\label{sec:lower-bound}
\begin{lemma}[\cite{vdam2}]
 Let $H(s)=(1-s)H_0+s H_P$ be a Hamiltonian on $N$ qubits where $H_0=-\sum_{i=1}^N \sigma_x^{(i)}$, and $H_P=\sum_{\x\in\{0,1\}^N} E_\x \proj{\x}$ with $E_{\x^0}=0$ for some $\x^0\in\{0,1\}^N$ and $1\leq E_\x\leq M=O(N)$ for all $\x\neq \x^0$. Then, the eigenvalue gap of $H(s)$ satisfies $\Delta(s)\geq 2^{-O(N\log N)}$ for all $0\leq s\leq 1$.
\end{lemma}
\begin{proof}
 Let us first consider the case $s\geq\frac{2N+1}{2N+2}$. Note that the diagonal elements of $H(s)$ are given by $H_{\x\x}=s E_\x$, and the non-diagonal elements satisfy $\sum_{\y\neq \x} H_{\x\y}=-(1-s)N$. Therefore, all Gershgorin circles have radius $(1-s)N$ and the circle around $H_{\x^0\x^0}=0$ is disjoint from the circles around $H_{\x\x}\geq s$ (for $\x\neq \x^0$) as soon as $(1-s)N<s-(1-s)N$, that is, $s>\frac{2N}{2N+2}$. In that case, the smallest eigenvalue lies in the circle around $H_{\x^0\x^0}=0$, and all others in the union of the other circles, so that we have for the eigenvalue gap $\Delta(s)=E_1(s)-E_0(s)\geq s-2(1-s)N\geq \frac{1}{2N+2}\geq\Omega(\frac{1}{N})$ for $s\geq\frac{2N+1}{2N+2}$.

Let us now consider the case $s\leq\frac{2N+1}{2N+2}$. Let $Q=\frac{-H(s)+sM\id}{1-s}=A+\lambda(M\id-H_P)$, where $A=-H_0$ is the adjacency matrix of the hypercube and $\lambda=\frac{s}{1-s}$. Note that since $E_\x\leq M$, all elements of this matrix are non-negative. By the mixing properties of the hypercube, all elements of $A^N$ are at least 1, and therefore this is also true for $Q^N$. This implies that the gap between the largest and second largest eigenvalue of $Q^N$ is at least $1$, that is,
$\mu_0^N-\mu_1^N\geq 1$, where $\mu_0$ and $\mu_1$ are the largest and second largest eigenvalue of $Q$. The eigenvalues $\mu_i$ of $Q$ are upper-bounded by the spectral radius $\rho(Q)$, which in term is upper bounded by the norm $\norm{Q}_1=\max_\x\sum_{\y}|Q_{\x\y}|$, so that $\mu_i\leq \lambda M+N$ for all $i$. Moreover, we have $\mu_0^N-\mu_1^N\leq (\mu_0-\mu_1)(\mu_0+\mu_1)^{N-1}$, so that
$$
\mu_0-\mu_1\geq\frac{1}{(\mu_0+\mu_1)^{N-1}}\geq\frac{1}{(2\lambda M+2N)^{N-1}}\geq \frac{1}{(4 MN+2M+2N)^{N-1}},
$$
where we used the fact that $\lambda\leq 2N+1$ when $s\leq\frac{2N+1}{2N+2}$. Finally, we have $$
\Delta(s)=(1-s)(\mu_0-\mu_1)\geq \frac{1}{(2N+2)(4 MN+2M+2N)^{N-1}}\geq 2^{-O(N\log N)}.
$$
\end{proof}
\end{document}

%% file: energy-diff-order4-2local.tex
\begingroup
  \makeatletter
  \providecommand\color[2][]{%
    \GenericError{(gnuplot) \space\space\space\@spaces}{%
      Package color not loaded in conjunction with
      terminal option `colourtext'%
    }{See the gnuplot documentation for explanation.%
    }{Either use 'blacktext' in gnuplot or load the package
      color.sty in LaTeX.}%
    \renewcommand\color[2][]{}%
  }%
  \providecommand\includegraphics[2][]{%
    \GenericError{(gnuplot) \space\space\space\@spaces}{%
      Package graphicx or graphics not loaded%
    }{See the gnuplot documentation for explanation.%
    }{The gnuplot epslatex terminal needs graphicx.sty or graphics.sty.}%
    \renewcommand\includegraphics[2][]{}%
  }%
  \providecommand\rotatebox[2]{#2}%
  \@ifundefined{ifGPcolor}{%
    \newif\ifGPcolor
    \GPcolortrue
  }{}%
  \@ifundefined{ifGPblacktext}{%
    \newif\ifGPblacktext
    \GPblacktexttrue
  }{}%
  \let\gplgaddtomacro\g@addto@macro
  \gdef\gplbacktext{}%
  \gdef\gplfronttext{}%
  \makeatother
  \ifGPblacktext
    \def\colorrgb#1{}%
    \def\colorgray#1{}%
  \else
    \ifGPcolor
      \def\colorrgb#1{\color[rgb]{#1}}%
      \def\colorgray#1{\color[gray]{#1}}%
      \expandafter\def\csname LTw\endcsname{\color{white}}%
      \expandafter\def\csname LTb\endcsname{\color{black}}%
      \expandafter\def\csname LTa\endcsname{\color{black}}%
      \expandafter\def\csname LT0\endcsname{\color[rgb]{1,0,0}}%
      \expandafter\def\csname LT1\endcsname{\color[rgb]{0,1,0}}%
      \expandafter\def\csname LT2\endcsname{\color[rgb]{0,0,1}}%
      \expandafter\def\csname LT3\endcsname{\color[rgb]{1,0,1}}%
      \expandafter\def\csname LT4\endcsname{\color[rgb]{0,1,1}}%
      \expandafter\def\csname LT5\endcsname{\color[rgb]{1,1,0}}%
      \expandafter\def\csname LT6\endcsname{\color[rgb]{0,0,0}}%
      \expandafter\def\csname LT7\endcsname{\color[rgb]{1,0.3,0}}%
      \expandafter\def\csname LT8\endcsname{\color[rgb]{0.5,0.5,0.5}}%
    \else
      \def\colorrgb#1{\color{black}}%
      \def\colorgray#1{\color[gray]{#1}}%
      \expandafter\def\csname LTw\endcsname{\color{white}}%
      \expandafter\def\csname LTb\endcsname{\color{black}}%
      \expandafter\def\csname LTa\endcsname{\color{black}}%
      \expandafter\def\csname LT0\endcsname{\color{black}}%
      \expandafter\def\csname LT1\endcsname{\color{black}}%
      \expandafter\def\csname LT2\endcsname{\color{black}}%
      \expandafter\def\csname LT3\endcsname{\color{black}}%
      \expandafter\def\csname LT4\endcsname{\color{black}}%
      \expandafter\def\csname LT5\endcsname{\color{black}}%
      \expandafter\def\csname LT6\endcsname{\color{black}}%
      \expandafter\def\csname LT7\endcsname{\color{black}}%
      \expandafter\def\csname LT8\endcsname{\color{black}}%
    \fi
  \fi
  \setlength{\unitlength}{0.0500bp}%
  \begin{picture}(7200.00,5040.00)%
    \gplgaddtomacro\gplbacktext{%
      \csname LTb\endcsname%
      \put(990,1495){\makebox(0,0)[r]{\strut{}2}}%
      \put(990,2315){\makebox(0,0)[r]{\strut{}4}}%
      \put(990,3136){\makebox(0,0)[r]{\strut{}6}}%
      \put(990,3956){\makebox(0,0)[r]{\strut{}8}}%
      \put(990,4776){\makebox(0,0)[r]{\strut{}10}}%
      \put(1122,484){\makebox(0,0){\strut{}0}}%
      \put(1852,484){\makebox(0,0){\strut{}25}}%
      \put(2581,484){\makebox(0,0){\strut{}50}}%
      \put(3311,484){\makebox(0,0){\strut{}75}}%
      \put(4040,484){\makebox(0,0){\strut{}100}}%
      \put(4770,484){\makebox(0,0){\strut{}125}}%
      \put(5499,484){\makebox(0,0){\strut{}150}}%
      \put(6229,484){\makebox(0,0){\strut{}175}}%
      \put(6958,484){\makebox(0,0){\strut{}200}}%
      \put(484,2740){\rotatebox{90}{\makebox(0,0){\strut{}$(E_{12}^{(4)})^2$}}}%
      \put(4040,154){\makebox(0,0){\strut{}$N$}}%
    }%
    \gplgaddtomacro\gplfronttext{%
      \csname LTb\endcsname%
      \put(2893,4338){\makebox(0,0)[r]{\strut{}75-th Percentile}}%
      \csname LTb\endcsname%
      \put(2893,4118){\makebox(0,0)[r]{\strut{}Mean}}%
      \csname LTb\endcsname%
      \put(2893,3898){\makebox(0,0)[r]{\strut{}Median}}%
      \csname LTb\endcsname%
      \put(2893,3678){\makebox(0,0)[r]{\strut{}25-th Percentile}}%
    }%
    \gplbacktext
    \put(0,0){\includegraphics{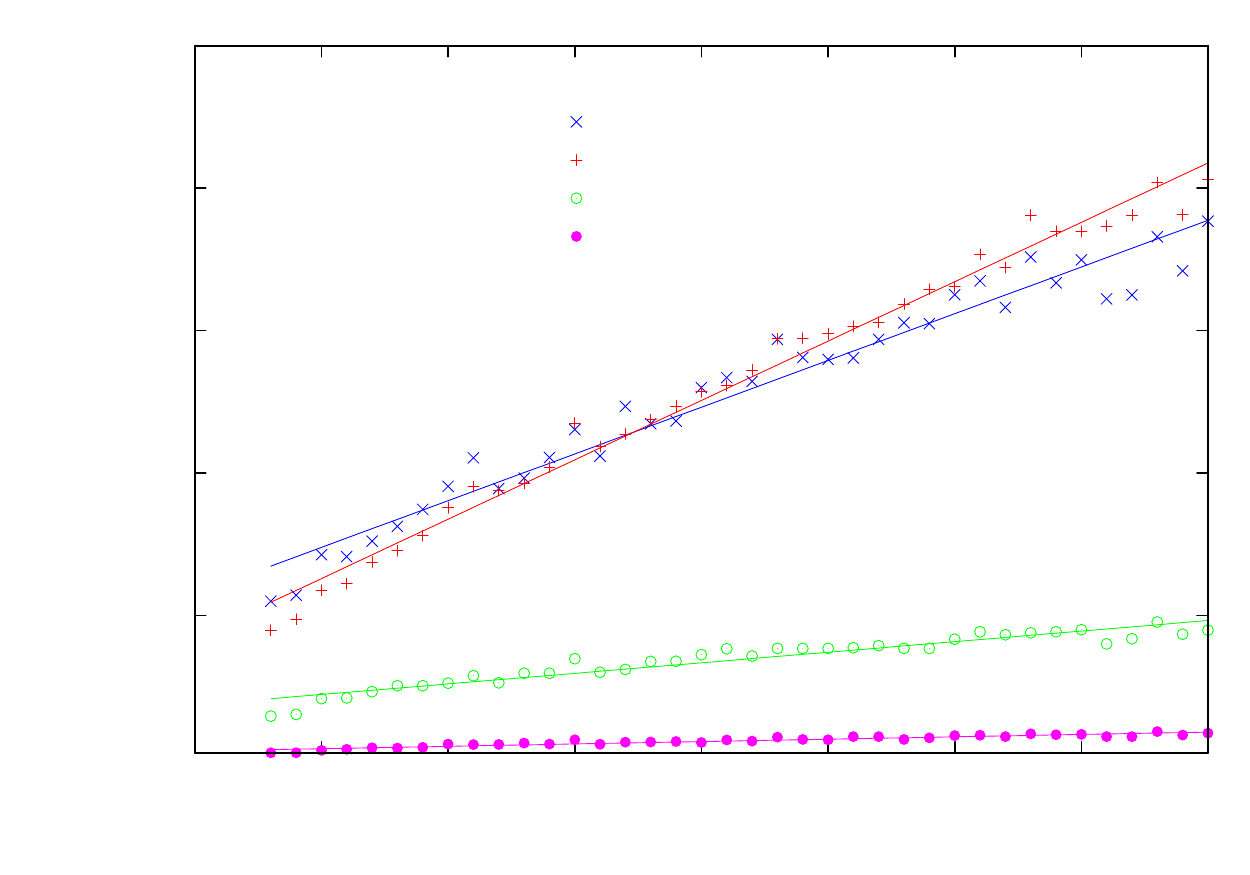}}%
    \gplfronttext
  \end{picture}%
\endgroup

%% file: energy-diff-order6-2local.tex
\begingroup
  \makeatletter
  \providecommand\color[2][]{%
    \GenericError{(gnuplot) \space\space\space\@spaces}{%
      Package color not loaded in conjunction with
      terminal option `colourtext'%
    }{See the gnuplot documentation for explanation.%
    }{Either use 'blacktext' in gnuplot or load the package
      color.sty in LaTeX.}%
    \renewcommand\color[2][]{}%
  }%
  \providecommand\includegraphics[2][]{%
    \GenericError{(gnuplot) \space\space\space\@spaces}{%
      Package graphicx or graphics not loaded%
    }{See the gnuplot documentation for explanation.%
    }{The gnuplot epslatex terminal needs graphicx.sty or graphics.sty.}%
    \renewcommand\includegraphics[2][]{}%
  }%
  \providecommand\rotatebox[2]{#2}%
  \@ifundefined{ifGPcolor}{%
    \newif\ifGPcolor
    \GPcolortrue
  }{}%
  \@ifundefined{ifGPblacktext}{%
    \newif\ifGPblacktext
    \GPblacktexttrue
  }{}%
  \let\gplgaddtomacro\g@addto@macro
  \gdef\gplbacktext{}%
  \gdef\gplfronttext{}%
  \makeatother
  \ifGPblacktext
    \def\colorrgb#1{}%
    \def\colorgray#1{}%
  \else
    \ifGPcolor
      \def\colorrgb#1{\color[rgb]{#1}}%
      \def\colorgray#1{\color[gray]{#1}}%
      \expandafter\def\csname LTw\endcsname{\color{white}}%
      \expandafter\def\csname LTb\endcsname{\color{black}}%
      \expandafter\def\csname LTa\endcsname{\color{black}}%
      \expandafter\def\csname LT0\endcsname{\color[rgb]{1,0,0}}%
      \expandafter\def\csname LT1\endcsname{\color[rgb]{0,1,0}}%
      \expandafter\def\csname LT2\endcsname{\color[rgb]{0,0,1}}%
      \expandafter\def\csname LT3\endcsname{\color[rgb]{1,0,1}}%
      \expandafter\def\csname LT4\endcsname{\color[rgb]{0,1,1}}%
      \expandafter\def\csname LT5\endcsname{\color[rgb]{1,1,0}}%
      \expandafter\def\csname LT6\endcsname{\color[rgb]{0,0,0}}%
      \expandafter\def\csname LT7\endcsname{\color[rgb]{1,0.3,0}}%
      \expandafter\def\csname LT8\endcsname{\color[rgb]{0.5,0.5,0.5}}%
    \else
      \def\colorrgb#1{\color{black}}%
      \def\colorgray#1{\color[gray]{#1}}%
      \expandafter\def\csname LTw\endcsname{\color{white}}%
      \expandafter\def\csname LTb\endcsname{\color{black}}%
      \expandafter\def\csname LTa\endcsname{\color{black}}%
      \expandafter\def\csname LT0\endcsname{\color{black}}%
      \expandafter\def\csname LT1\endcsname{\color{black}}%
      \expandafter\def\csname LT2\endcsname{\color{black}}%
      \expandafter\def\csname LT3\endcsname{\color{black}}%
      \expandafter\def\csname LT4\endcsname{\color{black}}%
      \expandafter\def\csname LT5\endcsname{\color{black}}%
      \expandafter\def\csname LT6\endcsname{\color{black}}%
      \expandafter\def\csname LT7\endcsname{\color{black}}%
      \expandafter\def\csname LT8\endcsname{\color{black}}%
    \fi
  \fi
  \setlength{\unitlength}{0.0500bp}%
  \begin{picture}(7200.00,5040.00)%
    \gplgaddtomacro\gplbacktext{%
      \csname LTb\endcsname%
      \put(1122,1357){\makebox(0,0)[r]{\strut{}20}}%
      \put(1122,2041){\makebox(0,0)[r]{\strut{}40}}%
      \put(1122,2725){\makebox(0,0)[r]{\strut{}60}}%
      \put(1122,3409){\makebox(0,0)[r]{\strut{}80}}%
      \put(1122,4092){\makebox(0,0)[r]{\strut{}100}}%
      \put(1122,4776){\makebox(0,0)[r]{\strut{}120}}%
      \put(1254,484){\makebox(0,0){\strut{}0}}%
      \put(1967,484){\makebox(0,0){\strut{}25}}%
      \put(2680,484){\makebox(0,0){\strut{}50}}%
      \put(3393,484){\makebox(0,0){\strut{}75}}%
      \put(4106,484){\makebox(0,0){\strut{}100}}%
      \put(4819,484){\makebox(0,0){\strut{}125}}%
      \put(5532,484){\makebox(0,0){\strut{}150}}%
      \put(6245,484){\makebox(0,0){\strut{}175}}%
      \put(6958,484){\makebox(0,0){\strut{}200}}%
      \put(484,2740){\rotatebox{90}{\makebox(0,0){\strut{}$(E_{12}^{(6)})^2$}}}%
      \put(4106,154){\makebox(0,0){\strut{}$N$}}%
    }%
    \gplgaddtomacro\gplfronttext{%
      \csname LTb\endcsname%
      \put(2966,4324){\makebox(0,0)[r]{\strut{}75-th Percentile}}%
      \csname LTb\endcsname%
      \put(2966,4104){\makebox(0,0)[r]{\strut{}Mean}}%
      \csname LTb\endcsname%
      \put(2966,3884){\makebox(0,0)[r]{\strut{}Median}}%
      \csname LTb\endcsname%
      \put(2966,3664){\makebox(0,0)[r]{\strut{}25-th Percentile}}%
    }%
    \gplbacktext
    \put(0,0){\includegraphics{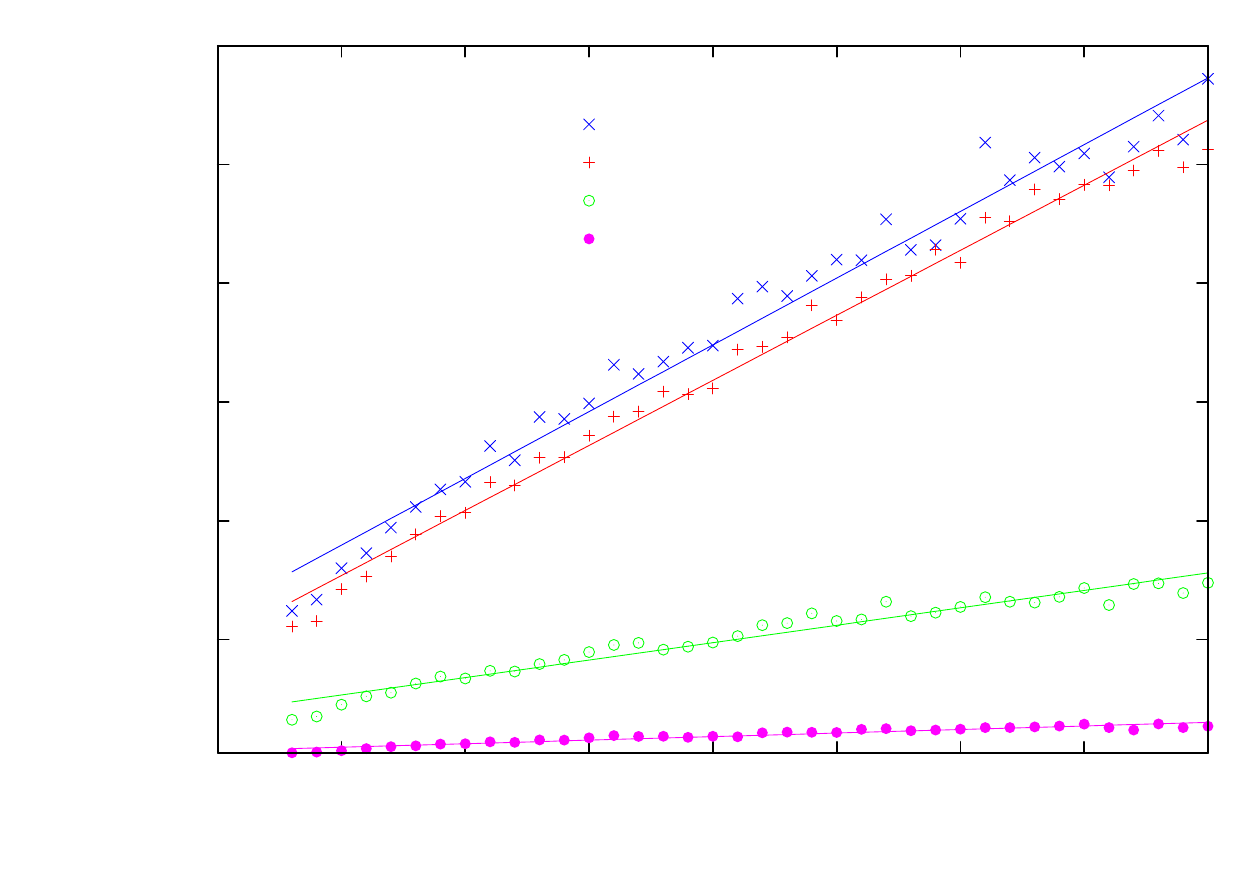}}%
    \gplfronttext
  \end{picture}%
\endgroup